\pdfoutput=1
\documentclass[10pt, conference, letterpaper]{IEEEtran}

\usepackage[utf8]{inputenc}

\usepackage{graphicx}
\usepackage{subcaption}
\usepackage{amsmath}
\usepackage{cleveref}
\usepackage{amssymb}
\usepackage{mathtools}
\usepackage{algpseudocode}
\usepackage{color}

\usepackage{txfonts}
\usepackage{bm} 

\graphicspath{ {Figures/} }
\usepackage{subcaption}

\usepackage[T1]{fontenc}
\usepackage{cite}
\usepackage{dsfont}

\makeatletter
\algnewcommand\Input{\item[\textbf{Input:}]}%
\algnewcommand\Output{\item[\textbf{Output:}]}%
\algnewcommand\Init{\item[\textbf{Init:}]}%
\makeatother

\usepackage{amstext} 
\usepackage{array}   
\newcolumntype{L}{>{$}l<{$}} 
\newcolumntype{C}{>{$}c<{$}} 
\newcolumntype{R}{>{$}r<{$}} 

\usepackage{amsthm}
\theoremstyle{plain}
\newtheorem{myTheorem}{Theorem}

\newtheorem{myLemma}{Lemma}

\theoremstyle{definition}
\newtheorem{myRemark}{Remark}

\newcommand{\setOfReals}{\mathbb{R}}
\newcommand{\setOfNaturals}{\mathbb{N}}
\newcommand{\setOfNonnegativeIntegers}{\mathbb{N}_0}
\newcommand{\setOfPositiveReals}{\setOfReals_{+}}
\newcommand{\setN}[1]{ [ #1]}

\newcommand{\Bin}[2]{\text{\underline{Binomial}} (#1, #2 )}

\newcommand{\PowerDistribution}[1]{  \text{\underline{Pow}} (#1 ) }

\DeclareMathOperator*{\argmin}{arg\,min}

\newcommand{\indicator}[1]{\mathds{1}(#1)}

\newcommand{\myExp}[1]{\exp \bigl( #1 \bigr)  }

\newcommand{\E}{\mathrm{E}}
\newcommand{\Eof}[1]{\E[#1]}

\newcommand{\prob}{\mathrm{P}}
\newcommand{\probOf}[1]{\prob(#1)}

\newcommand{\eqpunkt}{.}
\newcommand{\eqkomma}{,}
\newcommand{\defeq}{\coloneqq}
\newcommand{\disteq}{\, =_{\mathcal{D}} \, }

\newcommand{\ie}{\textit{i.e.}}
\newcommand{\eg}{\textit{e.g.}}

\usepackage{todonotes}
\usepackage{ifthen}

\newboolean{draftversion}
\setboolean{draftversion}{false} 

\newboolean{longVersion}
\setboolean{longVersion}{true}  

\newcommand{\proofLocation}[2]{\ifthenelse{\boolean{longVersion}}{\text{#1}}{\text{#2}}}

\newcommand{\tdWasiur}[2][]{\ifthenelse{\boolean{draftversion}}{\todo[inline, color=blue!20, caption={2do}, #1]{\begin{minipage}{\textwidth-4pt}\emph{Remark Wasiur:}\\#2\end{minipage}}}{}}

\newcommand{\tdAmr}[2][]{\ifthenelse{\boolean{draftversion}}{\todo[inline, color=orange!20, caption={2do}, #1]{\begin{minipage}{\textwidth-4pt}\emph{ToDo for Amr:}\\#2\end{minipage}}}{}}

\newcommand{\remarkAlex}[2][]{\ifthenelse{\boolean{draftversion}}{\todo[inline, color=red!20, caption={2do}, #1]{\begin{minipage}{\textwidth-4pt}\emph{Remark Alex:}\\#2\end{minipage}}}{}}

\newcommand{\tdGeneral}[2][]{\ifthenelse{\boolean{draftversion}}{\todo[inline, color=green!20, caption={2do}, #1]{\begin{minipage}{\textwidth-4pt}\emph{ToDO:}\\#2\end{minipage}}}{}}

\usepackage{url}

\usepackage{tabularx} 
\usepackage{xcolor}
\usepackage{xkeyval}
\input{tudcolours.def}

\usepackage{tikz}
\usetikzlibrary{positioning,fit,shapes,backgrounds,circuits.logic.US}

\tikzstyle{server}=[circle, line width=0.5pt, rounded corners=0.1mm, draw=black!100, fill=tud3a!100]
\tikzstyle{vertex}=[circle, line width=0.5pt, draw=black!100, fill=tud0a!100]
\tikzstyle{dispatcher} =[and gate US, line width=0.5pt, draw=black!100, fill=tud1a!100]
\tikzstyle{dotbox} = [draw=white, fill=white, rectangle,  inner sep=10pt, inner ysep=20pt]

\tikzset{three_sided/.style={
		draw=none,rectangle, 
		append after command={
			[shorten <= -0.5\pgflinewidth]
			([shift={(-1.5\pgflinewidth,-0.5\pgflinewidth)}]\tikzlastnode.north west)
			edge([shift={( 0.5\pgflinewidth,-0.5\pgflinewidth)}]\tikzlastnode.north east)
			([shift={( 0.5\pgflinewidth,-0.5\pgflinewidth)}]\tikzlastnode.north east)
			edge([shift={( 0.5\pgflinewidth,+0.5\pgflinewidth)}]\tikzlastnode.south east)
			([shift={( 0.5\pgflinewidth,+0.5\pgflinewidth)}]\tikzlastnode.south east)
			edge([shift={(-1.0\pgflinewidth,+0.5\pgflinewidth)}]\tikzlastnode.south west)
		}
	}
}

\IEEEoverridecommandlockouts

\title{Optimizing Stochastic Scheduling in Fork-Join Queueing Models: Bounds and Applications\ifthenelse{\boolean{longVersion}}{\\(Extended Version)}{}
}

\author{\IEEEauthorblockN{Wasiur R. KhudaBukhsh\IEEEauthorrefmark{1}, Amr Rizk\IEEEauthorrefmark{2}, Alexander Fr\"ommgen\IEEEauthorrefmark{2},  and Heinz Koeppl\IEEEauthorrefmark{1}}
	\IEEEauthorblockA{\IEEEauthorrefmark{1}Bioinspired Communication Systems Lab (BCS), E-Mail: \{wasiur.khudabukhsh \textbar~heinz.koeppl\}@bcs.tu-darmstadt.de,}
	\IEEEauthorblockA{\IEEEauthorrefmark{2}Multimedia Communications Lab (KOM), E-Mail: \{amr.rizk  \textbar~alexander.froemmgen\}@kom.tu-darmstadt.de, \\
		Technische Universitaet Darmstadt, Germany}
}

\begin{document}
\maketitle

\begin{abstract}

Fork-Join (FJ) queueing models 
capture the dynamics of system parallelization under synchronization constraints, for example, for applications such as MapReduce, multipath transmission and RAID systems.
Arriving jobs are first split into tasks and mapped to 
servers for execution, such that a job can only leave the system when all of its tasks are executed. 

In this paper, we provide computable stochastic bounds for the waiting and response time distributions for heterogeneous FJ systems under general parallelization benefit.
Our main contribution is a generalized mathematical framework for probabilistic server scheduling
strategies that are essentially characterized by a probability distribution over the number of utilized servers, and the optimization thereof. We highlight the trade-off  between the scaling benefit due to parallelization and the FJ inherent synchronization penalty. Further, we provide optimal scheduling strategies for arbitrary scaling regimes that map to different levels of parallelization benefit.
One notable insight obtained from our results is that different applications with varying parallelization benefits result in different optimal strategies. Finally, we complement our analytical results by applying
them
to various applications
showing the optimality of the proposed scheduling strategies.

\end{abstract}   \IEEEpeerreviewmaketitle


 \section{Introduction}
 \label{sec:introduction}
Fork-Join (FJ) queueing models naturally
capture the dynamics of system parallelization under synchronization constraints. They have 
seen a rise of interest as a modeling tool in the wake of massive improvement of the infrastructure for cloud computing and large-scale data processing.
The emergence of parallel data processing frameworks such as MapReduce \cite{dean2008mapreduce,polato2014comprehensive} and its implementation Hadoop \cite{Hashem2016} has significantly contributed to the modern IT infrastructure.
\remarkAlex{old version: as a modeling tool -> as modeling tool ... saved one line}

Fig.~\ref{fig:mapReduce} presents a MapReduce abstraction that closely resembles an FJ system.
Arriving jobs are first split into tasks each of which is then mapped exactly to one work-conserving server that executes the \emph{map} operation.
An optional \emph{combine} operation compresses the intermediate result to reduce the amount of data that is 
transferred through the network.
The compression efficiency depends on the application and, in particular, on the input the data size.
A job finally leaves the system when all of its tasks are executed.

In order to design better parallelized systems we require tractable models that connect system dynamics to corresponding key performance metrics. However, until today an exact analysis of FJ queueing systems in a general setup remains elusive \cite{baccelli1989fork,boxma1994queueing}.
It is particularly hard to find closed form expressions for the steady-state distributions of key quantities in FJ systems such as the waiting and response times.
In this paper, we contribute computable bounds for heterogeneous FJ systems under a fairly general setup.  Our main contribution is a generalized mathematical framework that allows the optimization of probabilistic server scheduling strategies that are shown to save server costs.

\begin{figure}
	\centering
	\includegraphics[width=1\linewidth]{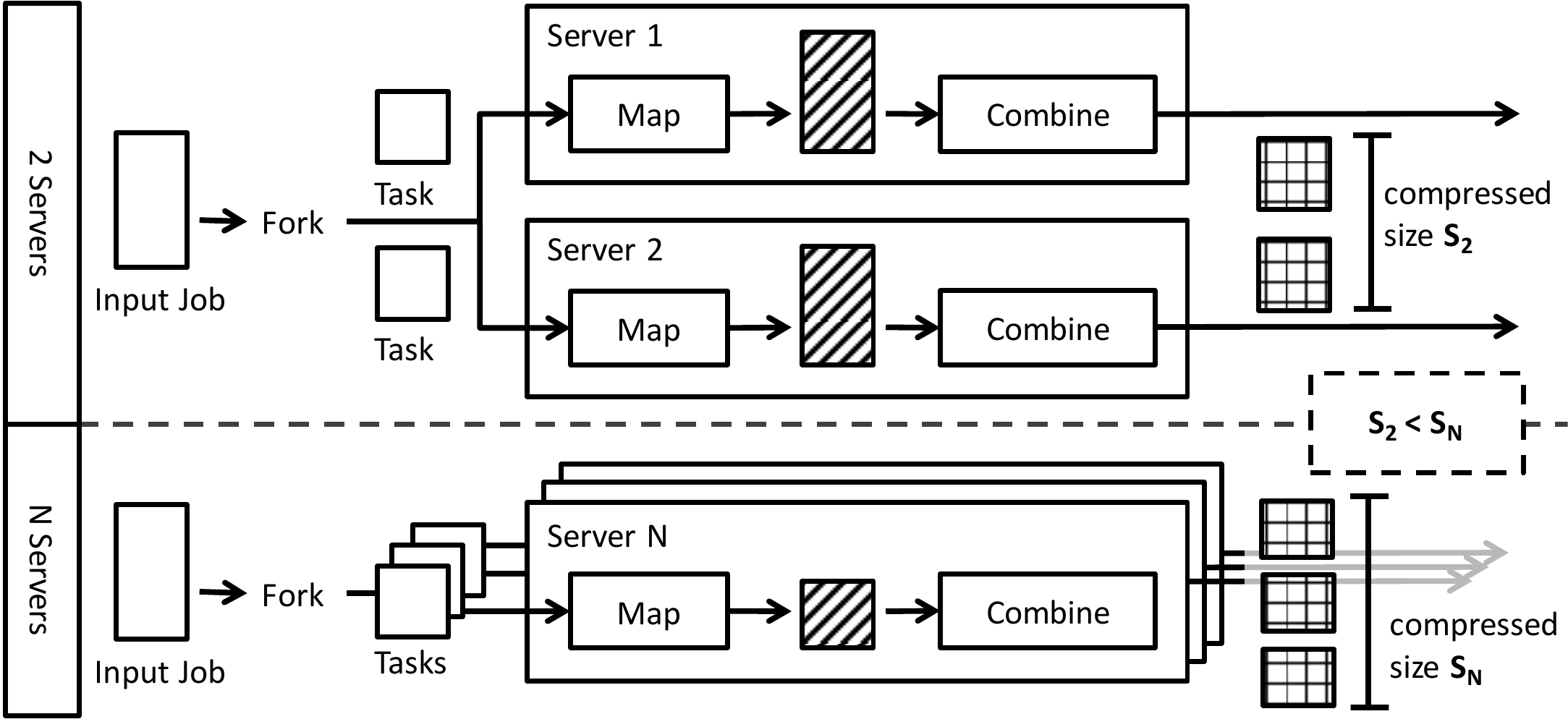}
	\caption{MapReduce as Fork-Join system: The output size of the \emph{combine}-phase may not scale linearly with the input size.}
	\label{fig:mapReduce}
	\vspace{-3pt}
\end{figure}

In this work, we model one of the main advantages of parallel systems, namely, the application specific parallelization benefit. To this end, we use the notion of service time scaling at each server of the FJ system. Since a job can only leave the system when all of its tasks are executed, we observe a naturally arising synchronization penalty in FJ systems.
In this paper, we analytically highlight this trade-off for arbitrary parallelization benefit regimes. We also show the impact of heterogenous servers on this trade-off.

Since in large pools of cloud resources, or, in general, in many parallelized systems, jobs are not mapped to \emph{all} available resources, and given the performance trade-off mentioned above, it is important to select the  number of utilized servers from a given pool of available ones in an informed way.
In the context of FJ systems, we define a scheduling strategy to be a probabilistic strategy of server selection. Clearly, a deterministic strategy is hence a degenerate case.
In this work, we formalize scheduling strategies in FJ systems, derive corresponding stochastic bounds on the waiting and response times, and minimize them to provide optimal strategies under arbitrary application specific parallelization benefits.

Our key contributions in this paper include: \textbf{(1)} Computable stochastic bounds for the steady-state distributions of the waiting and response times for a broad class of heterogeneous FJ systems for various scaling regimes.\footnote{We will use the terms \emph{scaling} and \emph{parallelization benefit} interchangeably.}  \textbf{(2)} A generalized mathematical framework for scheduling strategies that highlights the trade-off between parallelization benefit and the synchronization penalty, and enables finding optimal scheduling strategies for arbitrary scaling regimes. \textbf{(3)} Application of our model to different scenarios showing their efficiency.
\remarkAlex{It was: We show the strength of the model by applying our results to different application scenarios and thereby showing their efficiency.}

We organize the paper with a view to developing the concepts gradually and naturally, and to conveying the intuitions. Starting from the simplest case, we build up to the most general one.   The remainder of the paper is structured as follows: Sect.~\ref{sec:HetFJSystems} lays the mathematical foundation of our model of heterogeneous FJ systems. In Sect.~\ref{sec:job_assigment}, we introduce scheduling in FJ systems. Our main discussion on application specific scaling and scheduling under arbitrary scaling regimes is given in Sect.~\ref{sec:application_scaling}. In Sect.~\ref{sec:eval}, we consider concrete applications of our model and show corresponding findings. Finally, we discuss related work in Sect.~\ref{sec:related-work} and then conclude the paper with a short discussion in Sect.~\ref{sec:discussion}.

\section{Heterogeneous Fork-Join queueing systems}
\label{sec:HetFJSystems}

This section introduces FJ systems and provides stochastic bounds on the steady state waiting and response time distributions for a general heterogeneous setting. We denote the set of natural numbers by $\setOfNaturals$. Let $\setOfNonnegativeIntegers \defeq \setOfNaturals \cup \{0\}$. For an event $A$, $\indicator{A}$ is its indicator function.

\subsection{System description}
Consider a single stage FJ queueing system with $N$ parallel servers as depicted in Fig.~\ref{fig:mapReduce}. The servers are indexed on the set $\setN{N} \defeq \{1,2,\ldots, N \}$. Jobs arrive at the input station according to some point process with inter-arrival time $T_i$
between the $i$-th and $(i+1)$-th job, $i \in \setOfNaturals$. In the basic model a job is  split into $N$ tasks each of which is  assigned to exactly one server.
The service time for the task of job $i$ at the $n$-th server is denoted by the random variable $X_{n,i}$.
We shall assume independence of the families $\{X_{n,i}\}$ and $\{T_i\}$ throughout the course of this work. For lack of space, we only consider work-conserving servers in this paper.
We assume that the families $\{X_{n,i}\}$ and $\{T_i\}$ admit finite moment generating function (MGF) and Laplace transform,  defined as
$	\alpha_n(\theta)  \defeq  \Eof{e^{\theta X_{n,1}}}   \eqkomma
\beta(\theta)  \defeq  \Eof{e^{ - \theta T_{1}}} $,
respectively, for some $\theta >0 $
and for all $n \in \setN{N}$. 
We also assume the job arrival process is a renewal process.


\begin{figure*}[t!]
	\centering
	\begin{subfigure}[b]{0.32\textwidth}
		\centering
		\includegraphics[width=\textwidth]{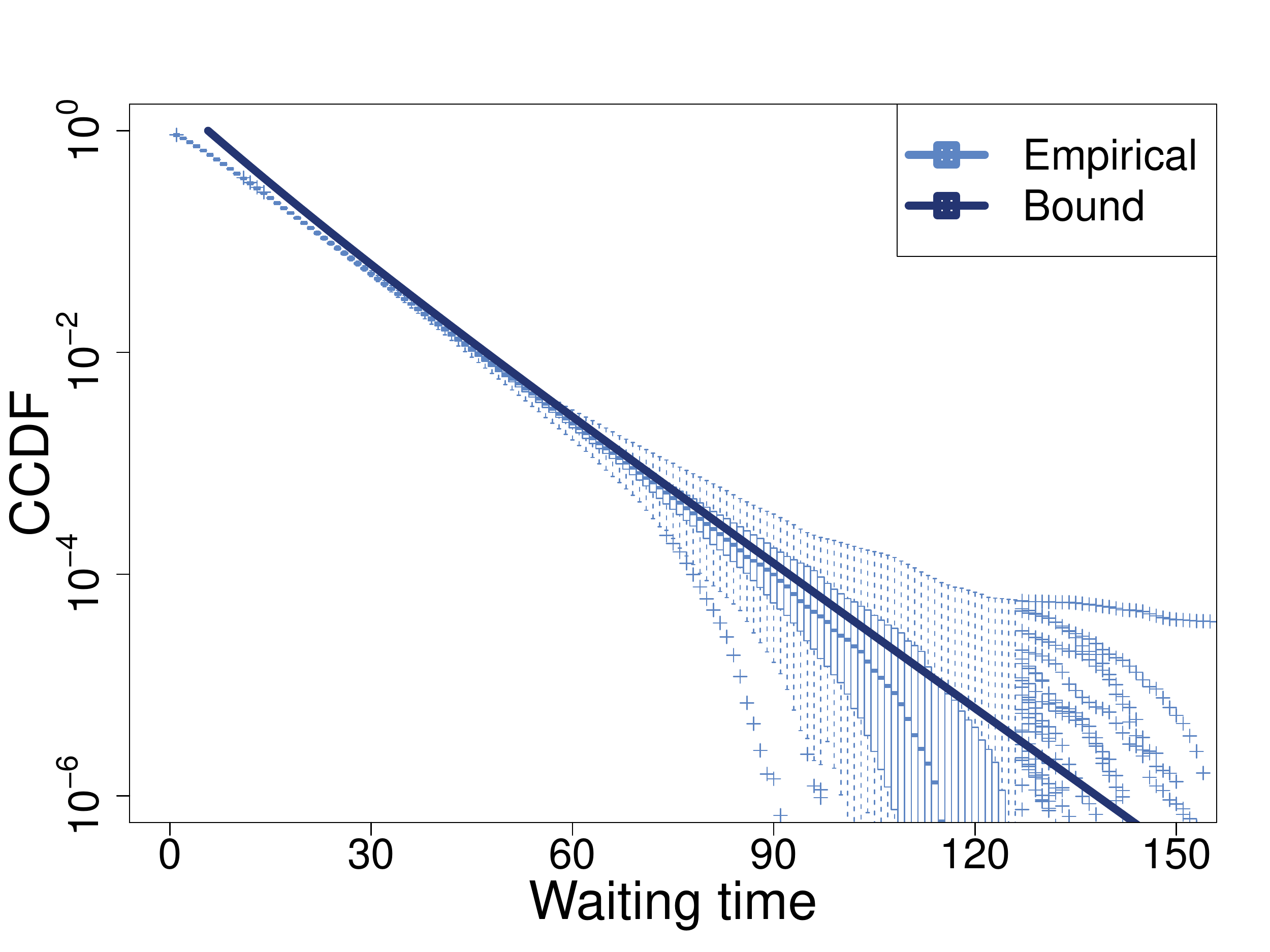}
	\end{subfigure}
	\begin{subfigure}[b]{0.32\textwidth}
		\centering
		\includegraphics[width=\textwidth]{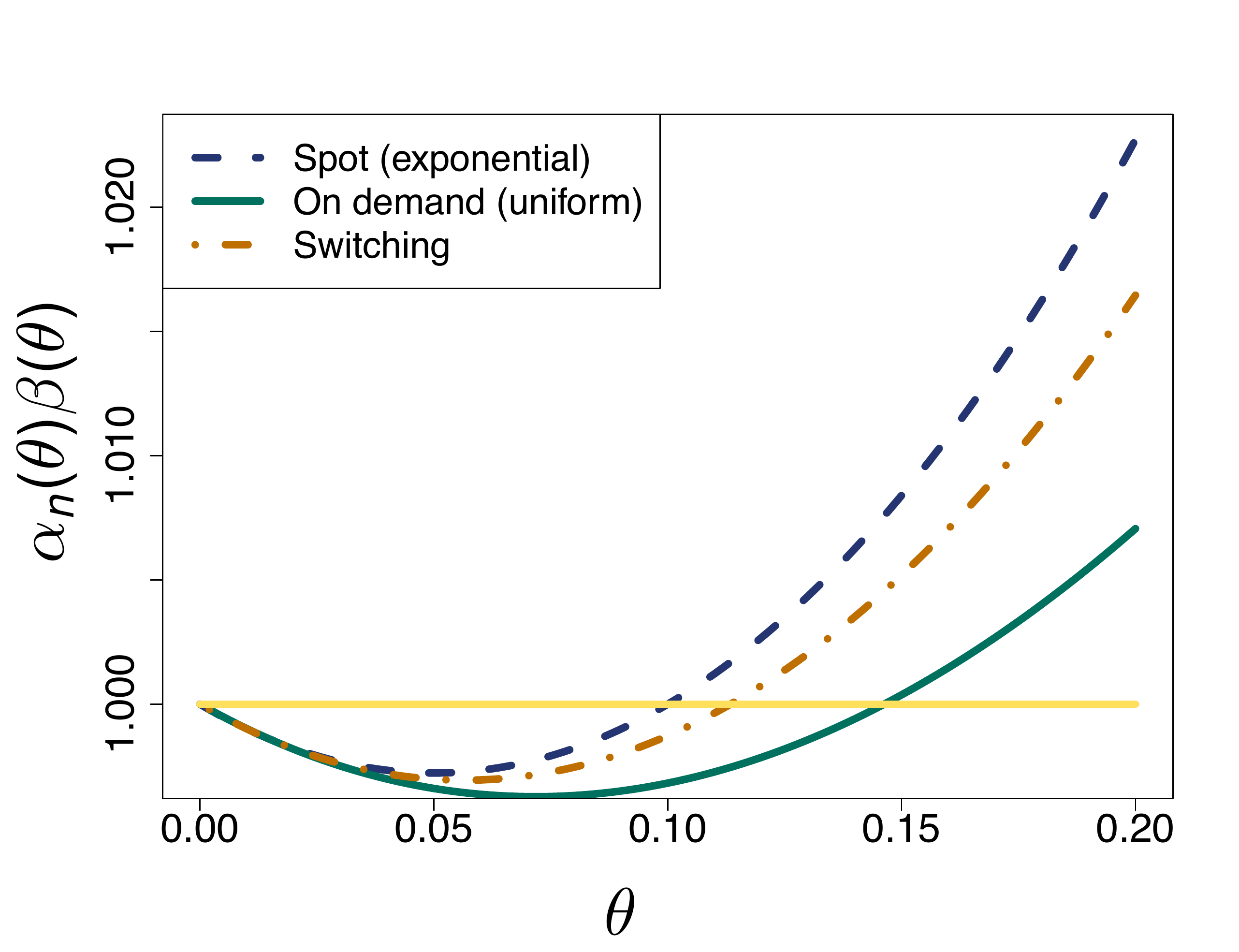}
	\end{subfigure}
	\begin{subfigure}[b]{0.32\textwidth}
		\centering
		\includegraphics[width=\textwidth]{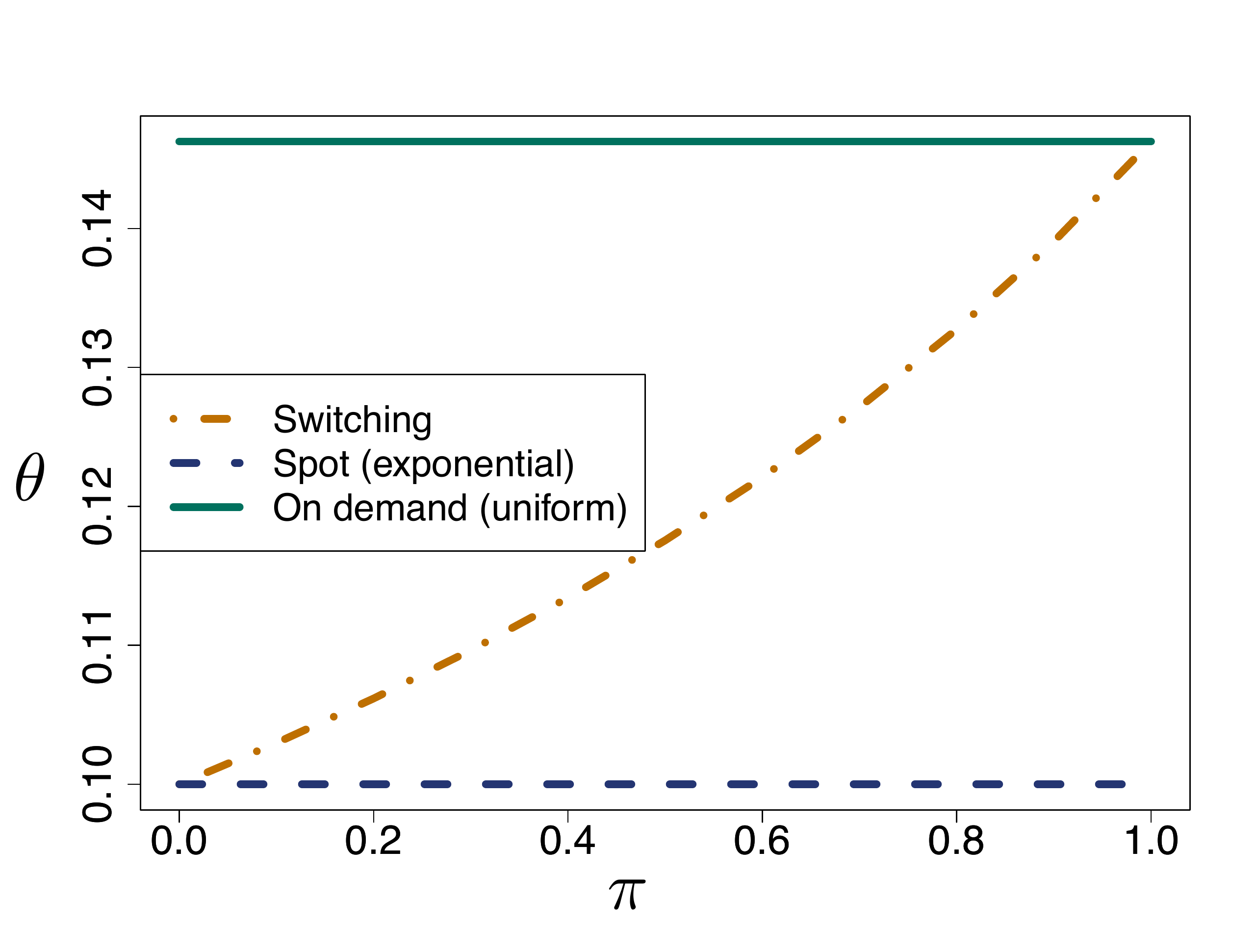}
	\end{subfigure}
	\vspace{-2mm}
	\caption{Example of a heterogeneous FJ system. (Left) Waiting time performance in a MapReduce cloud scenario with $N=2$ 
		partially volatile servers. One server is  on an  average faster representing a revocable checkpointed spot server with an exponential tail of service time. The second server provides on average slower service with uniformly distributed service times representing an on-demand server with stronger guarantees. The bound is calculated using Thm.~\ref{theorem:RenewalsNon-blocking}. CCDF denotes the complementary cumulative distribution function.  (Middle) The FJ system is constrained by the (on an  average) faster spot server due to its larger higher moments. This is apparent in the MGF condition $\alpha_n(x) \beta(x) =1$. Observe that the constraining decay rate is given by $ \tilde{\theta} \defeq \min_{n \in \setN{N} } \theta_n$. (Right) A system that switches between spot and on-demand servers with $\pi$ being the fraction of time where on-demand servers are used. Observe the improvement in the decay rate $\theta$ with increasing $\pi$. Simulation parameters: spot exponential service rate $\mu=1$, inter-arrival exponential rate $\lambda=0.9$ and uniform service time over $[0.001,2.009]$. }
	\vspace{-3mm}
	\label{fig:FJ_intro_example_spot}
\end{figure*}
\remarkAlex{changed Fork-Join system to FJ, saved a line in both columns}


\subsection{Waiting and response times for heterogeneous FJ Systems}
\label{sec:renewal_nonblocking}

In an FJ queueing system the waiting time $W_j$ is defined as $0$ for $j=1$ and $	\max\{ 0, \sup_{k \in \setN{j-1}} \{  \sup_{n \in \setN{N}} \{ \sum_{i=1}^{k} X_{n,j-i}  -\sum_{i=1}^{k} T_{j-i}  \} \}   \} \eqkomma $ for $j>1$ \cite{Rizk2015Sigmetrics}. Intuitively a job is considered to be waiting until its last task starts being serviced. The waiting time for the first job is assumed to be zero. 
Similarly the response time $R_j$ of job $j$  is defined as $ \max_{n \in \setN{N}} X_{n,1} $ for $j=1$ and $\sup_{k \in \setN{j-1}\cup \{0\} } \{  \sup_{n \in \setN{N}} \{ \sum_{i=0}^{k} X_{n,j-i}  -\sum_{i=1}^{k} T_{j-i}  \} \}  $ for $j>1$.
%
%
%
In order to get steady state representations of the above two random quantities, we require the stability condition $ \max_{n \in \setN{N}} \Eof{X_{n,1}} < \Eof{T_1} $.
%
%
Then, by  stationarity of the system, we have the following steady state representations of the waiting time $W$ and the response time $R$:
%
\begin{equation}
\begin{aligned}
W  \disteq  \sup_{k \in \setOfNonnegativeIntegers} \{  \sup_{n \in \setN{N}}  \{ \sum_{i=1}^{k} X_{n,i}  -\sum_{i=1}^{k} T_{i}   \} \}  \eqkomma \\
R \disteq  \sup_{k \in \setOfNonnegativeIntegers} \{  \sup_{n \in \setN{N}}  \{ \sum_{i=0}^{k} X_{n,i}  -\sum_{i=1}^{k} T_{i}   \} \} \eqkomma
\end{aligned}\label{eq:steady_state_renewal_nonblocking}
\end{equation}
%
%
where $\disteq$ denotes equality in distribution.
Now, we provide our first result giving stochastic bounds on the tail probabilities of $W$ and $R$ upon which we build the rest of the  paper.
\begin{myTheorem}
Consider an FJ system with $N$ parallel work-conserving servers fed by renewal job arrivals with inter-arrival times $T_i$, for $i \in \setOfNaturals$. Assuming iid service times $X_{n,i}$ and pairwise independence of the servers,
the steady state waiting and response time distributions are bounded by
 \begin{align*}
  \probOf{ W \geq \sigma }  \leq {}& \myExp{ - \tilde{\theta}  \sigma  }  \sum_{n \in \setN{N} } \myExp{ -  ( \theta_n  -  \tilde{\theta}   )\sigma }   
   \eqkomma \\
  \probOf{ R \geq \sigma }  \leq {} &    \myExp{ - \tilde{\theta}  \sigma  }  \sum_{n \in \setN{N} }  \alpha_n( \theta_n)  \myExp{ -  ( \theta_n  -  \tilde{\theta}   )\sigma }   
  \eqkomma
 \end{align*}
 where 
 $\theta_n$ is the positive solution of $ \alpha_n(x) \beta(x) =1$
 for $n \in \setN{N}$ and $\tilde{\theta} \defeq \min_{n \in \setN{N} } \theta_n$.
 \label[theorem]{theorem:RenewalsNon-blocking}
\end{myTheorem}

The  key steps involved in the proof of the above theorem are: 1) constructing separate martingales for each of the servers; 
and 2) applying Doob's sub- and supermartingale inequalities (see 
\cite{ash2014real}) to arrive at the  bounds. 
The detailed proof is provided in \ifthenelse{\boolean{longVersion}}{Sect. \ref{sec:Appendix}}{\cite{KhudaBukhsh2016techreport}}. 
%
Note that the stability condition 
guarantees the existence of $\theta_n >0$ such that $ \alpha_n( \theta_n) \beta(\theta_n) =1 $ for all $n \in \setN{N}$ (see \cite{poloczek2014scheduling, boxma1994queueing}). Hence, $\tilde{\theta} > 0$ is well defined.

\noindent\textbf{Example: Hedging using revocable cloud resources.} We consider a mixed cloud service consisting of both highly guaranteed and revocable resources. This service could be supplied by  infrastructure providers such as Amazon EC2~\cite{EC2}, or by a virtual provider on top using, e.g., on-demand or revocable spot market machines \cite{SubramanyaGSIS15}.

Consider an application of parallel computation under synchronization such as MapReduce \cite{EC2} or Spark \cite{Spark} requiring $N$ machines. In this example, we consider the case of exchanging on-demand machines with spot machines to save costs. In general, for a fixed budget the user obtains \emph{faster} spot machines in comparison to on-demand machines. The price difference arises naturally since spot machines are  at risk of revocation \cite{SubramanyaGSIS15}. We abstract the characteristics of these two classes of machines (\emph{on-demand} and \emph{spot}) through  different job service time distributions. Through revocation and application checkpointing procedures \cite{SubramanyaGSIS15} that are associated with spot machines, we generally model the tail of the corresponding job service time distributions to decay slower than in the case of on-demand machines.
For illustration we assume that the tail of the job service times decays exponentially in case of spot machines while in the case of on-demand machines we model the service times by a uniform distribution. Note that the following argument only requires that the tail of the service times decays slower for spot machines.

Fig.~\ref{fig:FJ_intro_example_spot} (left) shows the waiting time distribution in the case of exchanging an on-demand machine by an - on an average faster - spot machine. At first sight this seems to be a good idea, however, looking at Fig.~\ref{fig:FJ_intro_example_spot} (middle) we clearly see that the system is constrained by the spot machine which has lower average service time, however, a thicker tail. 
The figure on the right shows the utility of trading an on-demand machine with a spot one. While a greater usage of the on-demand machine incurs greater cost, it also increases the decay rate of the waiting and response times, $\theta$ which in turn leads to monetary saving due to faster job execution times.
\remarkAlex{Added left and middle after the figures,in particular, as we highlight the difference between these two figures}
\section{Scheduling tasks in heterogeneous FJ systems}
\label{sec:job_assigment}
%
%
In this section, we study basic scheduling mechanisms that decide on the number of servers to be used from a pool of available servers\footnote{Note that our notion of scheduling  differs from traditional scheduling algorithms such as the Shortest-Remaining-Processing-Time-first (SRPT).}.  Since in large pools of cloud resources (in general for parallelized systems) an arriving job is not scheduled on \emph{all }available resources, we
consider for  each server 
if it is selected to execute a task of an arriving job or not. 
Specifically, when a job arrives we consider that each server $n$ is selected with a probability $\pi_n$. This server selection probability $\pi_n$ can be used to model different aspects of parallelized systems, such  as the server failure rate  in  cloud computing facilities, a quality of service differentiation parameter for different applications, and a tuning parameter to control the degree of replication. Hence, different $\pi_n$ may exist for different classes of users.
Mathematically, the revised task service times $\tilde{X}_{n,i}$ are defined as $X_{n,i}$ with probability $\pi_n$ and $0$ with probability $1-\pi_n$.
The MGF of $ \tilde{X}_{n,i} $ is given by $\alpha_n^{*}(\theta) = (1- \pi_n) + \pi_n \alpha_n(\theta)  $.
The stability condition $\max_{n \in \setN{N}} \Eof{X_{n,i}} < \Eof{T_1}$
ensures the existence of the decay rate $\theta_n >0$ from Thm. \ref{theorem:RenewalsNon-blocking} for each $n \in \setN{N}$ such that $ \alpha_n^{*}(\theta_n) \beta(\theta_n) =1$. Define $\tilde{\theta} \defeq \min_{n \in \setN{N} } \theta_n >0 \eqpunkt$
We retain the same mathematical setup as before except for $X$ being replaced by $\tilde{X}$. 


\begin{myTheorem}
Consider an FJ system with $N$ parallel work-conserving servers fed by renewal job arrivals with inter-arrival times $T_i$, for $i \in \setOfNaturals$. The probability that the $n$-th server is selected  at the arrival of a job is $\pi_n$. Assuming iid service times $X_{n,i}$ and pairwise independence of the servers,
the steady state waiting and response time distributions are bounded by
\vspace{-1mm}
	\begin{align*}
	\probOf{ W \geq \sigma } \leq{} &  \myExp{ - \tilde{\theta}  \sigma  }  \sum_{n \in \setN{N} } \myExp{ -  ( \theta_n  -  \tilde{\theta}   )\sigma }  
	\eqkomma \\
	\probOf{ R \geq \sigma } \leq {} &   \myExp{ - \tilde{\theta}  \sigma  }  \sum_{n \in \setN{N} }  \alpha_n( \theta_n)  \myExp{ -  ( \theta_n  -  \tilde{\theta}   )\sigma }   
	 \eqkomma
	\end{align*}
	
	where 
	$\theta_n$ is the positive solution of $	\alpha^{*}_n(x) \beta(x) =1 \eqkomma$
	for $n \in \setN{N}$ and $\tilde{\theta} \defeq \min_{n \in \setN{N} } \theta_n$.
	\label[theorem]{theorem:PRMappingRenewalsNon-blocking}
\end{myTheorem}

\proofLocation{The proof is provided in Sect.~\ref{sec:Appendix}.}{The proof is provided in \cite{KhudaBukhsh2016techreport}.} 

%
%
%

\noindent\textbf{Example: Mixed server pool with different availability.}
Consider 
a pool of heterogeneous servers that are available according to some probability $\pi_i$. For simplicity,
we consider only three heterogeneous servers used for parallel processing. Note that this scenario can be easily generalized to $N$ servers using Thm.~\ref{theorem:PRMappingRenewalsNon-blocking}. For the sake of simplicity, we assume that the task service times are exponentially distributed with server specific rates $\mu_i$ and that jobs arrive according to some renewal process with exponentially distributed inter-arrival times with parameter $\lambda$. Note that the probability $\pi_i$ also signifies the fraction of time  server $i$ is used, hence, it is directly related to the computation cost in case of time priced resources.

Fig.~\ref{fig:RPM_No_Scaling} shows the change in the mean and the percentile of the waiting time due to the addition of  a server with a selection probability $\pi_i$ to a system of two permanently used servers each with $\pi_j=1$. For example, the lowest curve in Fig.~\ref{fig:RPM_No_Scaling} (left) shows the increase in the average waiting time if the slowest server is added with increasing probability $\pi_i$.
\remarkAlex{changed Fig a to left}

\noindent\textbf{Optimal Strategy.} It can be shown that the bound in Thm.~\ref{theorem:PRMappingRenewalsNon-blocking} is an increasing function of the number of servers $N$ and that the decay rate $\tilde{\theta}$ can be maximized,
\ie, the bound can be minimized by choosing only the the strongest server.




\begin{figure}
	\begin{subfigure}[b]{0.23\textwidth}
		\includegraphics[scale=0.17]{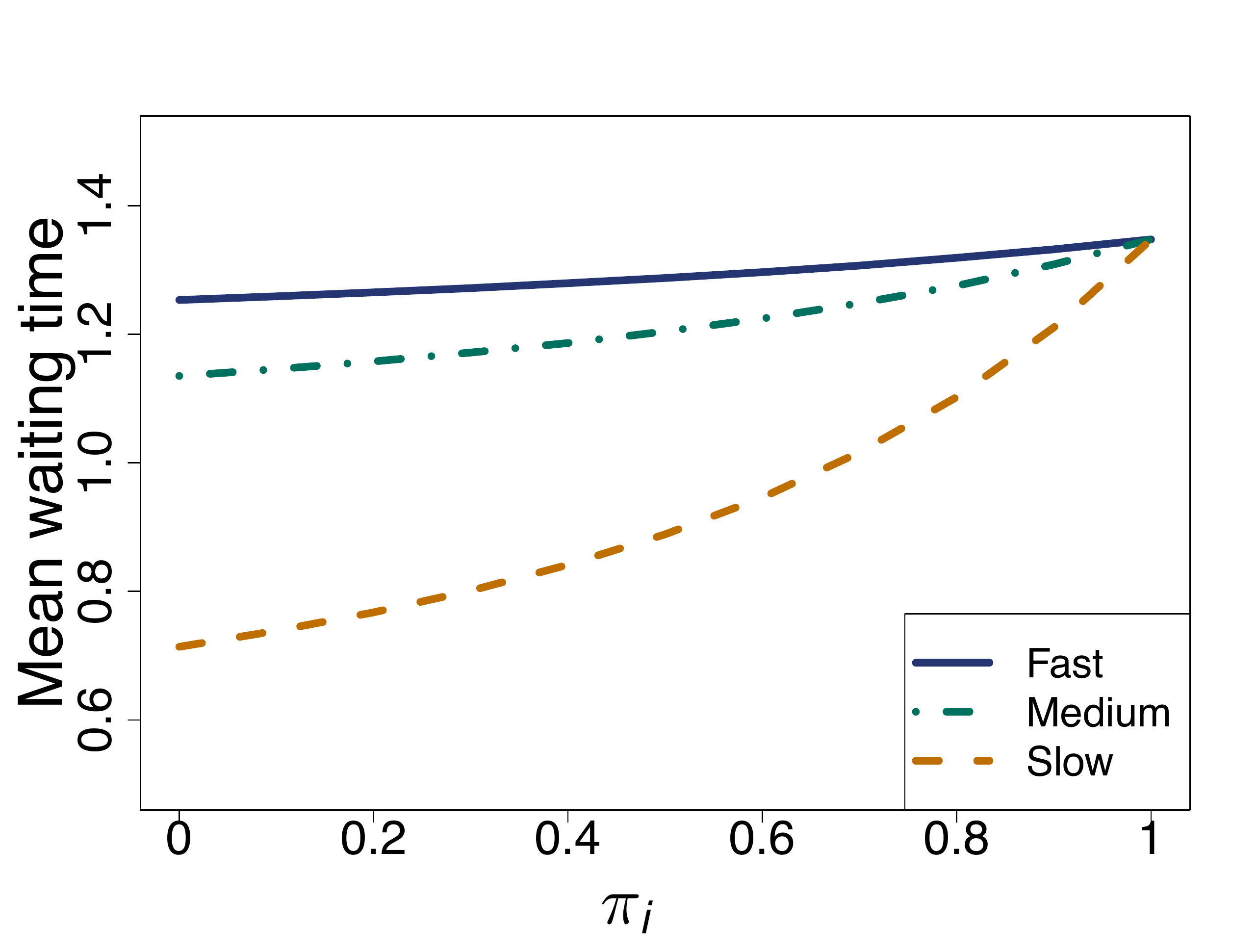}
	\end{subfigure}
	\begin{subfigure}[b]{0.23\textwidth}
		\includegraphics[scale=0.17]{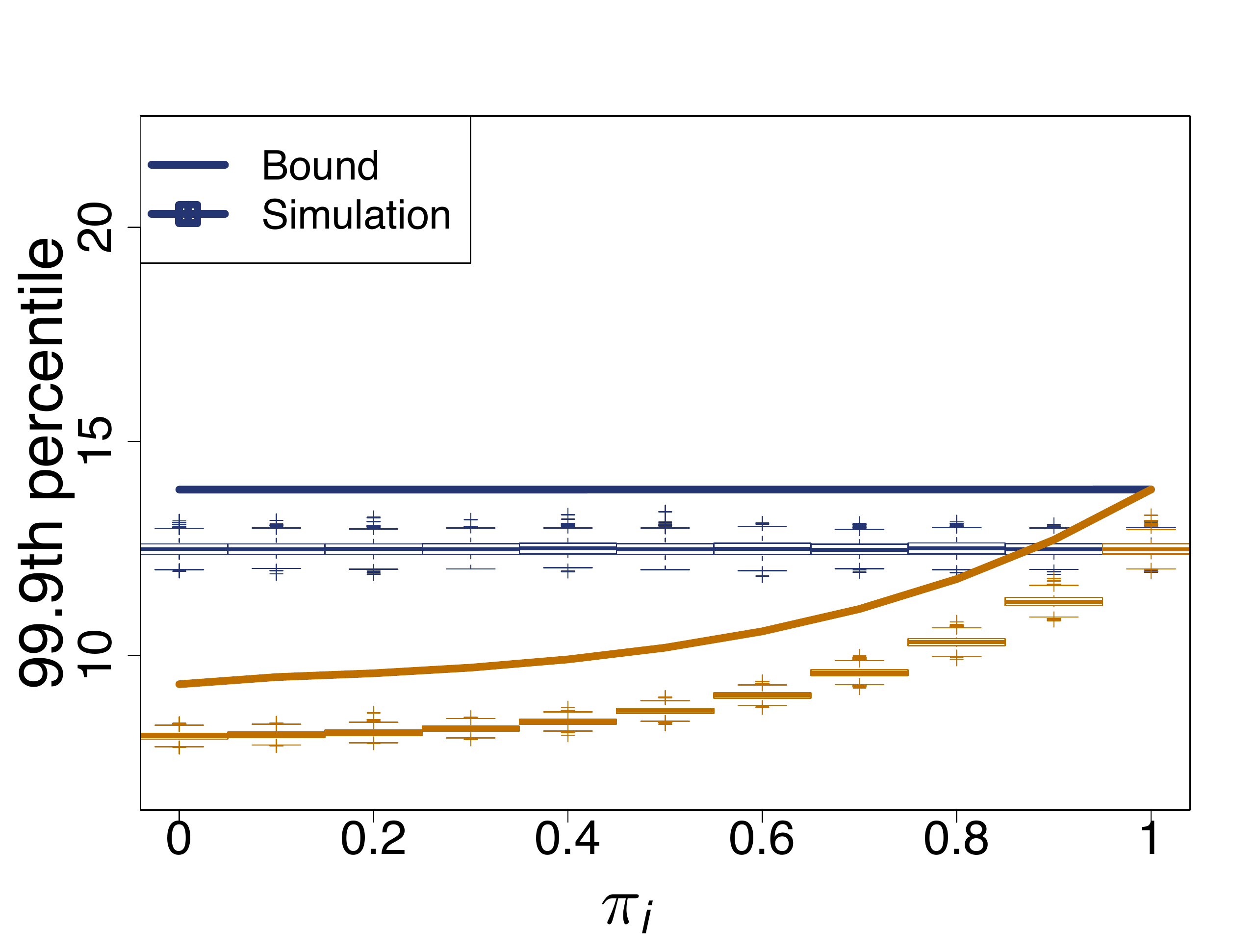}
	\end{subfigure}
		\caption{ 
			Impact of the degree of usage of a server on the mean (left) and the $99.9$-th percentile (right) of the steady-state waiting times. We consider a pool of three heterogeneous servers (fast, medium, slow), where tasks are always scheduled on two servers and the third server is included with probability $\pi_i$. Parameters:  service exponential rates $( \mu_1, \mu_2, \mu_3)  = (1.5,1.25,1)$  
			and inter-arrival exponential rate $\lambda=0.5$. 
		}
		\label{fig:RPM_No_Scaling}
		\vspace{-3pt}
\end{figure}



\section{Scheduling under application specific scaling}
\label{sec:application_scaling}

In this section, we analyze scheduling in FJ 
systems under application specific workloads. We build on the fact that different applications receive different gains from parallelization that lies in the nature of the application itself. Consider for example a Monte-Carlo simulation and a video transcoding application. In the first case, the gain from parallelization is strong and apparent, while in the second case, the gain from parallelization may vary depending on different factors such as the dependency between video macroblocks~\cite{chong2007efficient,mesa2009scalability}. We capture these varying gains using the notion of scaled service times. 
Moreover, in Fork-Join systems (e.g. MapReduce) there is a synchronization price that increases with the number of servers $N$ \cite{baccelli1989fork,Rizk2015Sigmetrics}.
We make the case that given these two opposing forces, the scheduling strategy that chooses the number of utilized  servers in an FJ system can be optimized to minimize the waiting and response times in the system. We begin with the initial case of homogeneous servers before discussing the more general case of heterogeneous servers.

%
%
%
%
\subsection{Homogeneous Servers - Linear scaling}
\label{sec:homogeneous_servers_linear}

The first natural scaling that we analyze is what we call \emph{linear scaling}\footnote{Linear scaling has been introduced in \cite{Rizk2016} for a fixed number of homogeneous servers~$N$ without considering scheduling strategies.}. This 
is motivated by examples of FJ systems where incoming jobs are equally divided among the servers.
%
%
Consider an FJ system with $N$ parallel, identical servers 
fed by renewal job arrivals with inter-arrival times $T_i$. We choose the servers probabilistically and once chosen, stick to them for a long time.
This allows us to write down steady-state representations conditional on the chosen set.
Let the random variable $S  \sim f_S \in  \mathcal{P} (\setN{N}, 2^{\setN{N}} ) $ denote the number of servers chosen to split an incoming job into, where $  \mathcal{P} (\setN{N}, 2^{\setN{N}} )$ is the class of all probability distributions on $ (\setN{N}, 2^{\setN{N}} )$\footnote{We use the symbol $2^A$ to denote the power set of a set $A$.}.
Let the service times at the $n$-th server $X_{n,i}$ be iid for all $i \in \setOfNaturals$ and $n \in \setN{N}$. Suppose the unscaled service time at each server is distributed as $X$, \ie, $ X_{n,i} \mid \{ S=1 \} \disteq X$
for some $ X$
with MGF $\alpha(x)$. We model the reduction of the average amount of work to be performed by each server when we use multiple servers using the following scaling of service times
\begin{equation}
X_{n,i} \mid \{ S=s \} \disteq \frac{X}{s} \eqpunkt
\label{eq:scaling_assumption}
\end{equation}
%

Now, conditional on the given number of used servers $ \{ S=s \}$ for some $s \in \setN{N}$, the steady-state waiting times $W$ and the response times $R$ can be represented as in \eqref{eq:steady_state_renewal_nonblocking} with $\setN{N}$ replaced by $\setN{s}$. We have the following result.
%
%
\begin{myTheorem}
Consider a stable FJ system with $N$ parallel work-conserving 
servers and renewal job arrivals with inter-arrival times $T_i$, for $i \in \setOfNaturals$. Let 
$S \sim  f_S \in  \mathcal{P} (\setN{N}, 2^{\setN{N}} ) $ denote the number of servers chosen to split an incoming job into. Let the unscaled service times $X$ and the inter-arrival times  $T$ be exponentially distributed with parameters $\mu$ and $\lambda$, respectively.
For service times $X_{n,i}$ at the $n$-th server that are scaled as in \eqref{eq:scaling_assumption}
	independently for all $ n \in \setN{S}, i \in \setOfNonnegativeIntegers $, the steady state waiting and response times are bounded as
	\begin{align*}
		\probOf{ W \geq \sigma  }  \leq {}& e^{\lambda \sigma }   \Eof{S e^{- \mu \sigma S}}~\eqkomma \\
		\probOf{ R \geq \sigma }  \leq {} &  \frac{ e^{  \lambda \sigma }  }{\rho} \Eof{ S^2 e^{  - \mu \sigma S}}~\eqkomma
	\end{align*}
where $\rho = \frac{\lambda}{\mu}$ is the unscaled utilization level and  the optimal strategy with respect to the bound for the waiting time 
	is
	\begin{align*}
	S_{opt} \sim f_{opt} ={}& \argmin_{  f_S \in  \mathcal{P} (\setN{N}, 2^{\setN{N}} )  }  \Eof{S e^{- \mu \sigma S}} \eqpunkt
	\end{align*}	
%
	\label[theorem]{theorem:RenewalsNon-blockingRPM_Scaled}
\end{myTheorem}
\proofLocation{The proof is provided in Sect.~\ref{sec:Appendix}.}{The proof is provided in \cite{KhudaBukhsh2016techreport}.} 
For a given choice of the distribution of $S$, which we call a \emph{strategy}, the bounds in Thm.~\ref{theorem:RenewalsNon-blockingRPM_Scaled} can be computed exactly, for it  involves only a summation of finitely many terms. Note that the optimization is essentially over 
a probability $N$-simplex
$\Delta_{N} \defeq \{ (p_1,p_2,\ldots,p_N) \in [0,1]^N \mid \sum_{k=1}^{N} p_k =1  \}$.

\noindent \textbf{Interpretation of the server selection strategy:} 
A strategy can be interpreted in two ways:
 \emph{(i)}
it actively arises through users' selection of 
different numbers of servers to utilize,  or 
\emph{(ii)}
it passively arises through a variable number of provided servers that are price volatile, e.g., spot instances at a given budget. In the following, we mainly take the former as an example for strategy derivations.

Note that different strategies lead to varying performance bounds, e.g., consider the case where we select the number of used servers uniformly at random from the pool of~$N$ servers, \ie, $\probOf{S=s} = (1/N ) \indicator{s \in \setN{N}}$. Then,  for $a>0$, 
$	\Eof{S e^{-aS}}  = \frac{e^{-a}}{N (1- e^{-a})}  [ \frac{1-  e^{ - (N+1)a } } {(1- e^{-a}) }
- (N+1) e^{- a N} ] $, and 
$	\Eof{S^2 e^{-aS}}  =    \frac{ e ^ {-2 a}}{N (1- e^{-a})  }  [    2 \frac{(1- e^{-(N+1)a} ) }{   (1- e^{-a}) ^2 }    
 - \frac{  2(N+1)e^{-Na}  - (1-  e^{ - (N+1)a })  }{ (1- e^{-a})  }  
- (N+1)  (  N e^{- (N-1)a} +  e^{- a N}  ) ] $. 
Setting $a=\mu \sigma$, closed-form expressions for the bounds in Thm.~\ref{theorem:RenewalsNon-blockingRPM_Scaled} are obtained. The uniform distribution allows little control over the number of selected servers.  To control the average number of utilized servers $\Eof{S}$ we
employ what we call a \emph{Binomial strategy}, i.e., we let $S$ follow a truncated binomial distribution on $\setN{N}$ with parameters $N$ and $p \in (0,1]$,
%
\begin{equation*}
\probOf{S=s } =  \frac{\binom{N}{s} p^s q^ {N-s}}{1 - q^N}  \indicator{s \in \setN{N}} \eqkomma
\end{equation*}
writing $q \defeq 1-p$. With abuse of notation, we write
$S \sim \Bin{N}{p}$. Given the total number of available servers~$N \in \setOfNaturals$, the binomial strategy allows us to vary $p$ to control the desired number of on average utilized servers $Np/ (1- q^N)$. 

Computing the expectations in Thm.~\ref{theorem:RenewalsNon-blockingRPM_Scaled} for $S \sim \Bin{N}{p}$, 
we get the following bounds
 \begin{align}
 \label{eq:bound_waitingtime_binomial}
 	\probOf{ W \geq \sigma  }  \leq {} &  N e^{- \theta \sigma} [ \frac{ p }{1 - q^N} ( p e^{- \mu \sigma} +q  )^{N-1}  ] \\
 	\probOf{ R \geq \sigma }	 \leq {} & \frac{N e^{-\theta \sigma } }{\rho } [ \frac{p  }{1 - q^N}  (  N  p e^{- \mu \sigma} + q  )   ( p e^{- \mu \sigma} +q  )^{N-2} ] \nonumber \eqpunkt
 \end{align}
%
\proofLocation{The proof is provided in Sect.~\ref{sec:Appendix}.}{The proof is provided in \cite{KhudaBukhsh2016techreport}.}

\noindent \textbf{Optimizing the Binomial strategy:} Our next goal  is to minimize the waiting times given a binomial strategy for server selection. Precisely, given $N$ available servers we look for $p$ that minimizes  the right hand side of \eqref{eq:bound_waitingtime_binomial} at some percentile~$\sigma$, \eg, the $99.9$-th percentile.
First, we rewrite the right hand side of \eqref{eq:bound_waitingtime_binomial} as $N e^{- \theta \sigma} \left[ ( \epsilon q +1 - \epsilon  )^{N-1}  /\sum_{k=0}^{N-1} q^k  \right]$
where we define $\epsilon \defeq 1- e^{- \mu \sigma }$. Next, we define $\psi : [0,1) \rightarrow \setOfPositiveReals $ as $ \psi(q) \defeq ( \epsilon q +1 - \epsilon  )^{N-1}  /\sum_{k=0}^{N-1} q^k  $ and study its behavior. Taking  derivative with respect to $q$, we get 
\begin{align*}
	\frac{\,d}{\,dq}\psi(q)
	={} & \frac{ ( \epsilon q +1 - \epsilon  )^{N-2}  }{  (\sum_{k=0}^{N-1} q^k  )^2 } \sum_{k=0}^{N-2}  ( N\epsilon -1-k) q^k  \eqpunkt
\end{align*}
Since $ ( \epsilon q +1 - \epsilon  )^{N-2}  /  (\sum_{k=0}^{N-1} q^k  )^2  >0$, the sign of the derivative is dictated by sign of the polynomial $\mathcal{Q}(q) \defeq \sum_{k=0}^{N-2}  ( N\epsilon -1-k) q^k $. Note that the coefficients $ \{ N\epsilon -1-k \}_{k \in  \{0\}  \cup \setN{N-2}}$ of the polynomial are monotonically decreasing, implying there is only one change of sign of the coefficients so that by \emph{Descartes' rule of signs}, there is at most one real root of $\mathcal{Q}(q) =0$. Consequently, the same holds true for $\frac{\,d}{\,dx}\psi(x) $.
Now, observe that $ \mathcal{Q}(0) =N\epsilon-1 >0$ if $ \epsilon > 1/N$. On the other hand, 
$ 	\mathcal{Q}(1) =  N (N-1)(\epsilon -1/2) >0 $ if $ \epsilon > 1/2  \iff \sigma  >    (1/\mu) \ln(  2).$ This condition on the $99.9$-th percentile of the waiting time holds except for corner cases with nearly no queueing. 
%
This gives us a sufficient condition for $ \frac{\,d}{\,dx}\psi(x) >0$ implying that $\psi(q)$ is an increasing function of $q$ on $ \epsilon > 1/2$. In other words, the tail bound is a decreasing function of $p$. \emph{Therefore, the optimal strategy would be to set $p_{opt}=1$ and use all $N$ available servers  to make the most of the scaling benefit.}
%
%
Our analytic arguments are also numerically validated using simulations, e.g., Fig.~\ref{fig:Binom_scaling_Percentiles_2}. 

\noindent\textbf{Optimization under budget constraint:}
In the  interesting scenario of an application with a budget constraint on the average number of servers it uses, the above reduces to a constrained optimization problem. Precisely, if we have a budget constraint of the form $\Eof{S} \leq S^*$, the optimization problem can be stated as
\begin{align*}
\min N e^{- \theta \sigma} [ \frac{ p }{1 - q^N} ( p e^{- \mu \sigma} +q  )^{N-1}  ] {}& \text{ \quad s. t. } \frac{Np}{1- q^N} \leq  S^* \eqkomma 
\end{align*}
leading to $p^*  = \sup \{ p \in (0,1] \mid \sum_{k=0}^{N-1} (1-p)^k \geq \frac{N }{S^*  } \} $ so that $f_{opt}= \Bin{N}{p^*} $, 
In general, the given bound can always be numerically optimized for any~$\sigma$.

\noindent\textbf{Generalization to Power series strategies:} To obtain bounds in the more general setup of a power series strategy, we assume
\begin{equation}
\probOf{S=s} \defeq \frac{ a_s \kappa^ s }{\zeta(\kappa)}   \indicator{s \in \setOfNaturals} \eqkomma
\end{equation}
where $\zeta(\kappa) \defeq \sum_{k \in \setOfNaturals} a_k \kappa^k < \infty $ for some $\kappa >0$ and $a_k \geq 0 ~ \forall k \in \setOfNaturals$. We denote this distribution by $\PowerDistribution{ \kappa, \zeta}$  
%
and the corresponding bounds on the waiting and response time distributions in Thm.~\ref{theorem:RenewalsNon-blockingRPM_Scaled} evaluate to
	\begin{align*}
		\probOf{ W \geq \sigma  }  \leq {} & e^{\lambda \sigma }  \frac{ \kappa e^{- \mu \sigma  } \zeta'(\kappa  e^{- \mu \sigma  } )  }{\zeta(\kappa)} \eqkomma \\
		\probOf{ R \geq \sigma }  \leq {} &  \frac{ e^{  \lambda \sigma }  }{\rho} \frac{ \kappa e^{-\mu \sigma  }   }{\zeta(\kappa)} [ \kappa e^{-\mu \sigma  }  \zeta''(\kappa  e^{-\mu \sigma  } ) +   \zeta'(\kappa  e^{-\mu \sigma } ) ]  \eqpunkt
	\end{align*}
\proofLocation{The proof is provided in Sect.~\ref{sec:Appendix}.}{The proof is provided in \cite{KhudaBukhsh2016techreport}.}  For a given form of $\zeta$, the strategy can be optimized to minimize the waiting times. We skip this optimization due to the lack of space. Please note that the above is the \emph{most} general result of this kind.
\subsection{Homogeneous Servers - Partial scaling}
\label{sec:homogen_partial}
In the previous subsection we considered  
linear scaling of the form \eqref{eq:scaling_assumption} that models a perfect work division over $s$ utilized servers in the sense of $\Eof{X_{n,i}} = \Eof{X}/s$.
In this section, we analyze the general case of application specific scaling,
i.e., where the parallelization benefit due to using more servers depends on the application itself. Two prominent examples are: \emph{(i)} MapReduce scenarios where the servers have to separately calculate a state before starting the task executions, and \emph{(ii)} parallelized video transcoding, where some involved decoding operations have a diminishing return on parallelization \cite{chong2007efficient,mesa2009scalability}.



Mathematically, we assume that for a certain application with scaling coefficient $\varphi \in [0,1]$, the following scaling down of service times holds,
\begin{equation}
X_{n,i} \mid \{ S=s \} \disteq \frac{X}{s^ \varphi } \eqpunkt
\label{eq:power_scaling_assumption}
\end{equation}

Given $ \{ S=s \}$, the steady-state waiting times $W$ and the response times $R$ have the same representation as in \eqref{eq:steady_state_renewal_nonblocking} where we need to replace~$N$ with $s$.
%
%
Now, we present our bounds in the partial scaling regime.
\begin{myTheorem}
	Consider a stable FJ system with $N$ parallel work-conserving 
servers and renewal job arrivals with inter-arrival times $T_i$, for $i \in \setOfNaturals$. Let the random variable $S \sim f_S \in  \mathcal{P} (\setN{N}, 2^{\setN{N}} ) $ denote the number of servers chosen to split an incoming job into. Let the unscaled service times $X$ and the inter-arrival times  $T$ be exponentially distributed with parameters $\mu$ and $\lambda$, respectively.
For service times $X_{n,i}$ at the $n$-th server that are scaled as in \eqref{eq:power_scaling_assumption} for some $\varphi \in [0,1]$
%
%
%
the steady state waiting and response times are bounded as
	\begin{align*}
		\probOf{ W \geq \sigma  }  \leq {} & e^{\lambda \sigma }   \Eof{S \myExp{- \mu \sigma S^ \varphi}} \eqkomma \\
		\probOf{ R \geq \sigma }  \leq {} &  \frac{ e^{  \lambda \sigma }  }{\rho} \Eof{ S^2 \myExp{  - \mu \sigma S^ \varphi}}  \eqkomma
	\end{align*}
	where $\rho = \frac{\lambda}{\mu}$ is the unscaled utilization level. The optimal strategy with respect to the bound for the waiting time 
	is
	\begin{align*}
	S_{opt} \sim f_{opt} ={}& \argmin_{  f_S \in  \mathcal{P} (\setN{N}, 2^{\setN{N}} )  }  \Eof{S e^{- \mu \sigma S^ \varphi} } \eqpunkt
	\end{align*}
	%
	\label[theorem]{theorem:PowerScaling}
\end{myTheorem}
\proofLocation{\noindent The proof is provided in Sect.~\ref{sec:Appendix}.}{\noindent The proof is provided in \cite{KhudaBukhsh2016techreport}.}
%
%
%
%
\begin{figure}
	\begin{subfigure}[b]{0.24\textwidth}
			\includegraphics[scale=0.47]{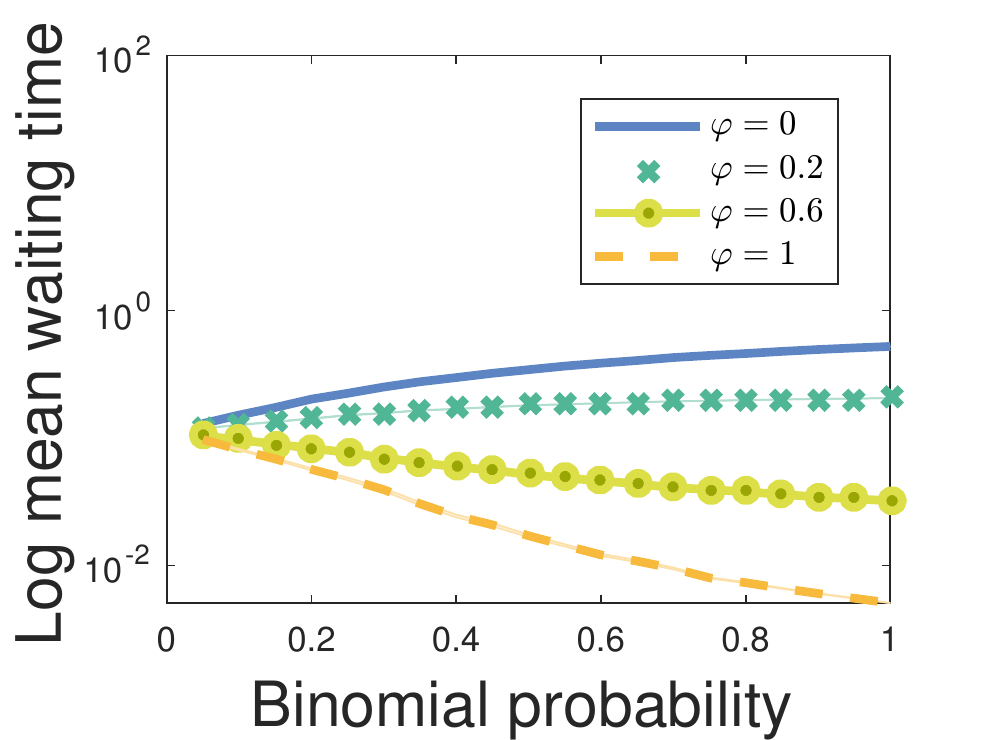}
	\end{subfigure}	\hfill
	\begin{subfigure}[b]{0.24\textwidth}
			\includegraphics[scale=0.47]{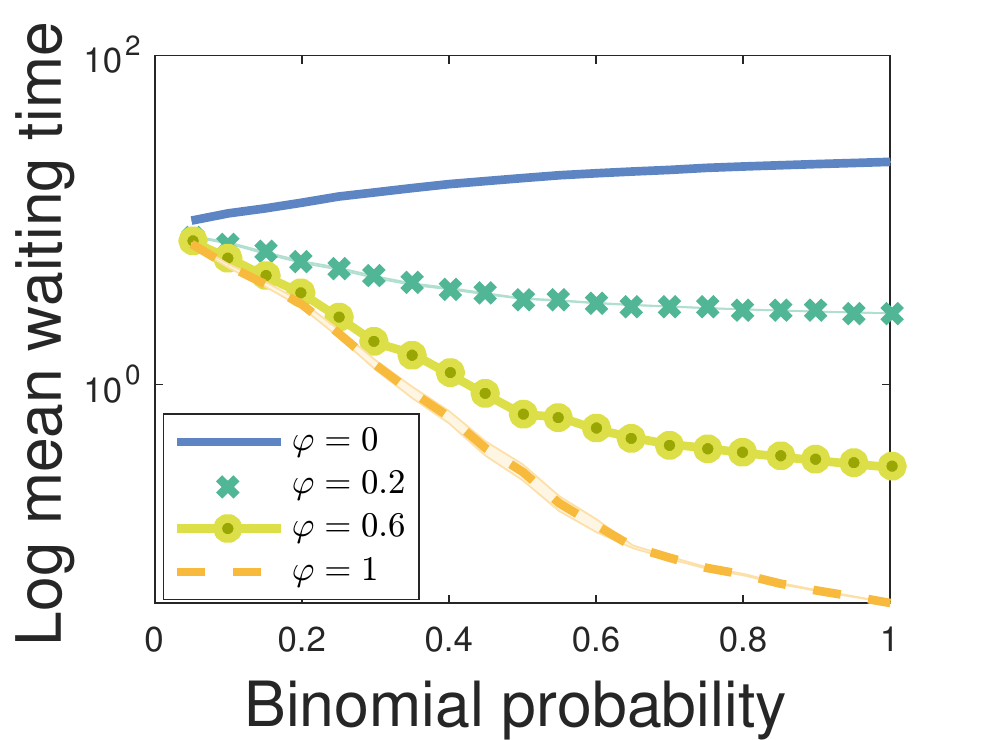}
	\end{subfigure}
	\caption{
The impact of the scheduling strategy (given by probability $p$) together with the parallelization benefit (given by increasing $\varphi$) on the mean waiting time in given FJ systems.
\newline Simulation parameters: $N=10$ servers, (Left) low utilization: $\lambda=0.1$. (Right) high utilization: $\lambda=0.9$.  }
\label{fig:w_vs_binom_p_diff_phi}
\end{figure}

\noindent\textbf{Insights into partial parallelization benefit:} Fig.~\ref{fig:w_vs_binom_p_diff_phi} conveys multiple insights into scheduling strategies under different application specific scaling $\varphi$. 
It depicts the mean waiting time in a given FJ system for different scheduling strategies given by the Binomial probability $p$ for various parallelization benefits given by the coefficient $\varphi$. 
%
The first insight from Fig.~\ref{fig:w_vs_binom_p_diff_phi} is the trade-off between the FJ inherent synchronization penalty and the parallelization benefit due to scaled service times.
For a given scheduling strategy in an FJ system, i.e., the probability $p$, we observe a decrease in the mean waiting time with increasing scaling benefit $\varphi$. Second, for low parallelization benefit $\varphi$, the synchronization penalty predominates leading to an increase in mean waiting times. We note that this phenomenon also depends on the utilization.
Finally, for high parallelization benefit $\varphi$, we observe a
decay of the mean waiting times with $p$, i.e., essentially increasing the average number of utilized servers $N p / (1-q^N)$. We observe a general diminishing behavior with $p$. Hence, for larger $\varphi$ substantial savings in server cost can be obtained by sacrificing a little in terms of the average waiting time. Fig.~\ref{fig:Binom_scaling_Percentiles_2} shows a similar behavior for the percentiles of the waiting time distribution.

Remarkably, we find that for any fixed stochastic strategy, i.e., $p \in (0,1]$ under no 
parallelization benefit, 
the percentiles of the waiting times grow as $\mathcal{O}(\log \Eof{S} )$. In case of no stochastic scheduling, i.e., $p=1$, we recover the behavior of $\mathcal{O}(\log N )$  known from \cite{baccelli1989fork,Rizk2015Sigmetrics}.

\noindent\textbf{Optimal strategy under partial scaling:} The prime motive of the analysis above 
is to gain analytic insights into the impact of the chosen number of servers on the waiting times 
for an application with a given scaling $\varphi$ in a fixed FJ system. In particular, given a $\varphi \in [0,1]$, we 
find the optimal stochastic scheduling strategy by minimizing the bound obtained in Thm.~\ref{theorem:PowerScaling}.  
%
%
Observe that as $\varphi \rightarrow 0$, the scaling benefit diminishes to zero yielding the unscaled case from Sect.~\ref{sec:HetFJSystems}. 
Further, as $\varphi \rightarrow 1$, we get greater scaling benefit. The optimal strategy, therefore, would be to choose all the servers if the scaling benefit outweighs the synchronization cost, and to choose only the strongest server if it does not. However, this depends on the parallelization benefit $\varphi$ specific to 
the given application. 
\begin{figure}
	\centering
	\includegraphics[width=0.5\textwidth]{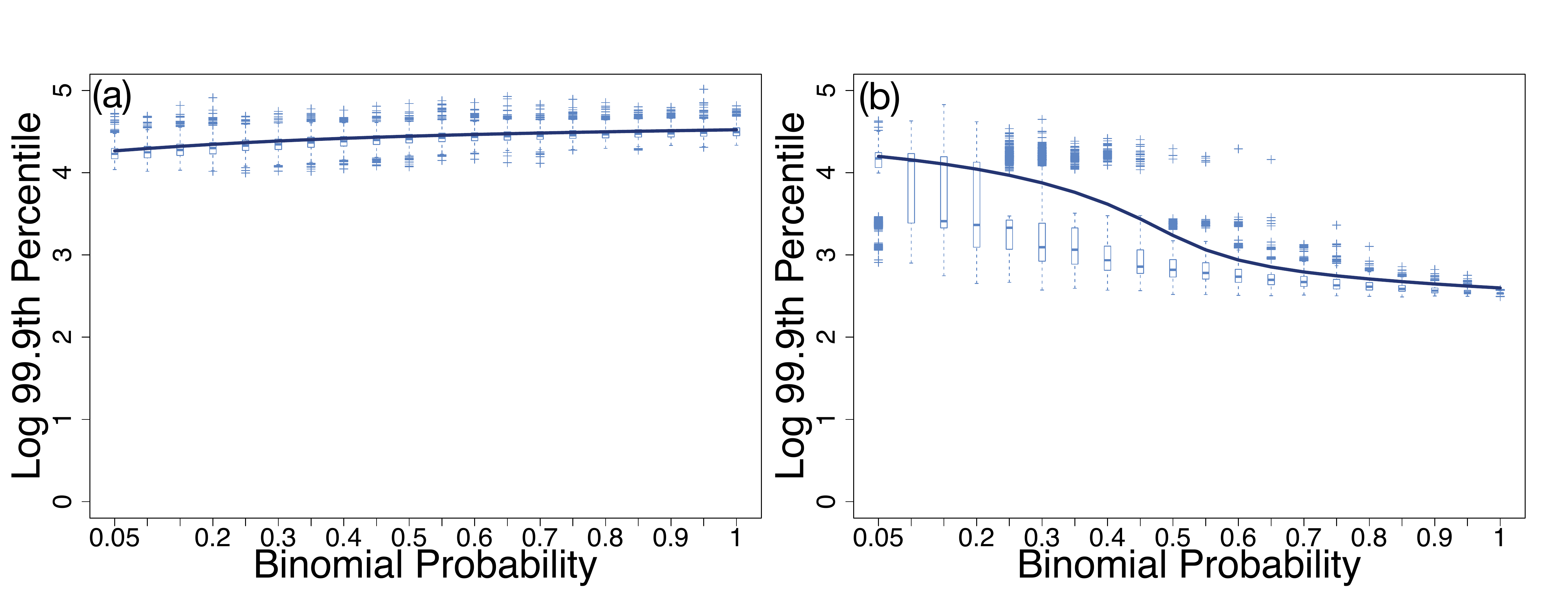}
	\caption{The impact of the scheduling strategy on the waiting time percentiles. Simulation parameters: $N=10, \lambda=0.9$, parallelization benefit: (a) $\varphi=0$ (b) $\varphi=0.2$. 
		}
\label{fig:Binom_scaling_Percentiles_2}
\end{figure}
\subsection{Heterogeneous Servers - Hierarchical Model}
In this section, we generalize our scaling discussion to the heterogeneous case,  building on the analytic intuitions gained in the previous section. We argue that the average service times at different servers are not  identical, but rather follow some suitable probability distribution (see Fig.~\ref{fig:HierarchicalModel}). 
Here, we assume  a randomly drawn server has an exponential service rate with parameter $\mu$ where $\mu$ itself is drawn from an underlying hierarchical distribution~$f_\mu$. 
We present the following result for such a setup, assuming the strict stability $\max_{n \in \setN{N}} \Eof{X_{n,1}} < \Eof{T_1}$.

%

\begin{myTheorem}
	Consider an FJ system with $N$ parallel work-conserving 
	servers fed by renewal job arrivals with iid exponentially distributed inter-arrival times $T_i$ with parameter $\lambda$, for $i \in \setOfNaturals$. Let the random variable $S \sim f_S \in  \mathcal{P} (\setN{N}, 2^{\setN{N}} )  $ denote the number of servers chosen to split an incoming job into and the unscaled service time $X_n$ at the $n$-th server be exponentially distributed with parameter 
$\mu_n \sim f_\mu$
For  service times $X_{n,i}$ at the $n$-th server that are scaled as
%
%
	\begin{equation*}
	X_{n,i} \mid \{ S=s \} \disteq \frac{X_n}{s^\varphi }  \eqkomma
	\end{equation*}
independently for all $ n \in \setN{s}, i \in \setOfNonnegativeIntegers $, $\varphi \in [0,1]$, the
 steady state waiting and response times are bounded as
	\begin{align*}
	\probOf{ W \geq \sigma  }  \leq {}& e^{\lambda \sigma }   \Eof{S \myExp{- \min_{n \in \setN{S}} \mu_n \sigma S ^\varphi  }}  \eqkomma \\
	\probOf{ R \geq \sigma }  \leq {} &  \frac{e^{\lambda \sigma }  }{\lambda} \Eof{ S^\varphi ( \sum_{n \in \setN{S} } \mu_n )  \myExp{- \min_{n \in \setN{S}} \mu_n \sigma S^\varphi }  }   \eqpunkt
	\end{align*}	
	The optimal strategy with respect to the bound above for the waiting time 
	is given by
	\begin{align*}
	S_{opt} \sim f_{opt} ={}& \argmin_{  f_S \in  \mathcal{P} (\setN{N}, 2^{\setN{N}} )  }    \Eof{S \myExp{- \min_{n \in \setN{S}} \mu_n \sigma S ^\varphi  }}  \eqpunkt
	\end{align*}

	\label[theorem]{theorem:scaling_in_heterogeneous}
\end{myTheorem}
\proofLocation{The proof is provided in Sect.~\ref{sec:Appendix}.}{The proof is provided in \cite{KhudaBukhsh2016techreport}.}

\noindent \textbf{Example: A two-class system:} Consider the case where there are only two types of servers in the system, \emph{fast} and \emph{slow}. In a cloud computing infrastructure, these two types would correspond to different monetary prices.
Suppose the exponential service rates of the two types of servers are $\kappa_1$ and $\kappa_2$, respectively, and the arrival rate is $\lambda$
with $\lambda < \kappa_1 < \kappa_2$. Denote the probability that a randomly drawn server is of type-$1$, \ie, has exponential service rate $\kappa_1$, by $\pi$. Hence, the service rate distribution is given by
\begin{equation}
	f_\mu(x) \defeq \pi^{\indicator{x= \kappa_1}  }   (1- \pi)^{ \indicator{x= \kappa_2} }  \eqpunkt
	\label{eq:mu_degree_twoType}
\end{equation}
Given $n$ random samples $\mu_1,\mu_2, \ldots, \mu_n$  from the above distribution, we require  
the first order statistic of the sample 
$Y_n \defeq \min_{i \in \setN{n}} \mu_i $ 
 to compute the bounds in Thm.~\ref{theorem:scaling_in_heterogeneous}. 
The distribution of $Y_n$ is given by 
$	\probOf{Y_n= \kappa_1} = 1- (1-\pi)^n = 1- \probOf{Y=\kappa_2}$, 
such that 
its MGF is
$	\Eof{e^{aY}} 
= \myExp{  a \kappa_1} - (    \myExp{  a \kappa_1} - \myExp{  a \kappa_2} ) (1-\pi)^n$, 
%
whence we can compute the bounds obtained in Thm.~\ref{theorem:scaling_in_heterogeneous} for different choices of distributions of the number of used servers~$S$. In particular, when $S \sim \Bin{N}{p}$ and we receive linear scaling $\varphi=1$, the upper bounds on the tail probabilities can be explicitly written  as
\begin{align*}
	\probOf{ W \geq \sigma  }  \leq {} &   e^{\lambda \sigma }   \frac{N p    }{1- q^N}  b_1(\sigma) [ 1-   (1-\pi)  (   \frac{c_1(\sigma)  - c_2(\sigma)   }{ b_1(\sigma) } ) ]  \eqkomma
\end{align*}
where $b_i(\sigma) \defeq \myExp{- \sigma   \kappa_i  } (p \myExp{-\sigma   \kappa_i } + q  ) ^{N-1} $  and $ c_i(\sigma) \defeq \myExp{- \sigma   \kappa_i  } (p  (1-\pi)  \myExp{-\sigma   \kappa_i     } + q  ) ^{N-1} $ for $i=1,2$. 
%
%

While the above example only considers two types of servers, it is worth mentioning that it can easily be extended to take into account finitely many types of servers.

\noindent \textbf{The hierarchical hyper-parameter model:}
%
In view of the stability of the system, we take $f_\mu$ to be a truncated exponential with (hyper-) parameter $\mu_0$, truncated at $\lambda$. That is, we take
\begin{equation}
f_\mu (x) \defeq \mu_0 \myExp{ - \mu_0 (x- \lambda) } \indicator{x > \lambda} \eqpunkt
\label{eq:hierarchical}
\end{equation}
Given $n$ random samples $\mu_1,\mu_2, \ldots, \mu_n$  from the above distribution, 
the first order statistic of the sample $Y_n \defeq \min_{i \in \setN{n}} \mu_i $ 
 has a truncated exponential distribution with parameter $n \mu_0$, truncated at $\lambda$. 
 The MGF of $Y_n$ is given as 
\begin{align*}
	\Eof{e^ {a Y_n}}
	={}& \frac{ n \mu_0    }{n \mu_0    -a } \myExp{ a \lambda}   \eqpunkt
\end{align*}
Taking the same approach as in Sect.~\ref{sec:homogen_partial},
we can compute the waiting and response time bounds from Thm.~\ref{theorem:scaling_in_heterogeneous} for different choices of distributions of~$S$. 
In particular for the linear scaling case, \ie,  $\varphi=1$ and when 	$S \sim \Bin{N}{p}$, the upper bounds on the tail probabilities can be explicitly found  as
\begin{align*}
	\probOf{ W \geq \sigma  }  \leq {} &   \frac{  Np \mu_0  }{ (1- q^N)  ( \mu_0  + \sigma ) }  (p e^{-\sigma \lambda } + q  ) ^{N-1} \eqpunkt
\end{align*}
\proofLocation{The proof is provided in Sect.~\ref{sec:Appendix}.}{The proof is provided in \cite{KhudaBukhsh2016techreport}.}

\noindent \textbf{Heterogeneous FJ systems - Three forces:}
As shown above the hierarchical model extends our findings in the previous sections to a wide setting providing insights and lending greater applicability.
Thm.~\ref{theorem:scaling_in_heterogeneous} shows that \emph{(i)} the first order statistic $Y_s \defeq \min_{i \in \setN{s}} \mu_i $ is decisive for the overall performance 
of the system, in addition, to the opposing forces from Sect.~\ref{sec:homogeneous_servers_linear}, i.e., \emph{(ii)} scaling of service times at each server due to the parallelization, and \emph{(iii)} the synchronization penalty at the output. In fact, the heterogeneous case provides less scaling benefit than the homogeneous case due to $Y_s$. This impact can be directly seen from the position of $Y_s$ in the exponent in Thm.~\ref{theorem:scaling_in_heterogeneous}. The optimal strategy given all the relevant parameter values is obtained, as before, by optimizing the upper bound provided in  \Cref{theorem:scaling_in_heterogeneous}.

\begin{figure}
	\centering
			\begin{tikzpicture}[inner sep=0cm, minimum size = 0.65cm,circuit logic US, huge circuit symbols]
			\node (server1)[server, fill= tud1a ] (1) at (0,  0)  {$f_\mu$} ;
			\node (server2)[server,  below left=1.5cm of 1] (2)  { $\mu_1$ } ;
			\node (server3)[server, right=0.75cm of 2] (3)  {$\mu_2$ } ;
			\node (dotbox)[dotbox, right= 0.3cm of 3](5){$\cdots$} ;
			\node (server4)[server, right=1.5cm of 3] (4)  { $\mu_s$ } ;
			
			\draw [-latex, line width=0.5pt, in=90, out=270] (1) to node [minimum size=0.3cm, xshift=0.1cm, yshift=0.25cm] { } (2);
			\draw [-latex, line width=0.5pt, in=90,out=270] (1) to node [minimum size=0.3cm, xshift=0.1cm, yshift=0.25cm] {} (3);
			\draw [-latex, line width=0.5pt, in=90,out=270] (1) to node [minimum size=0.3cm, xshift=0.1cm, yshift=0.25cm] {     $ \{ S=s  \}$ } (4);
			\end{tikzpicture}
		\caption{The hierarchical model for the heterogeneous FJ systems. Conditional on $\{S=s\}$, the average service rates are drawn from a hierarchical distribution~$f_\mu$. }
		\label{fig:HierarchicalModel}				
\end{figure}
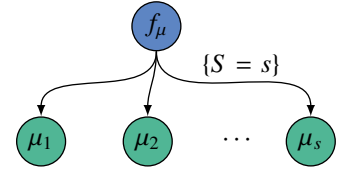

\section{Evaluation of Application Specific Scheduling in Fork-Join Systems}
\label{sec:eval}

\begin{figure}
	\centering
	\begin{subfigure}[b]{0.23\textwidth}
		\includegraphics[width=\textwidth]{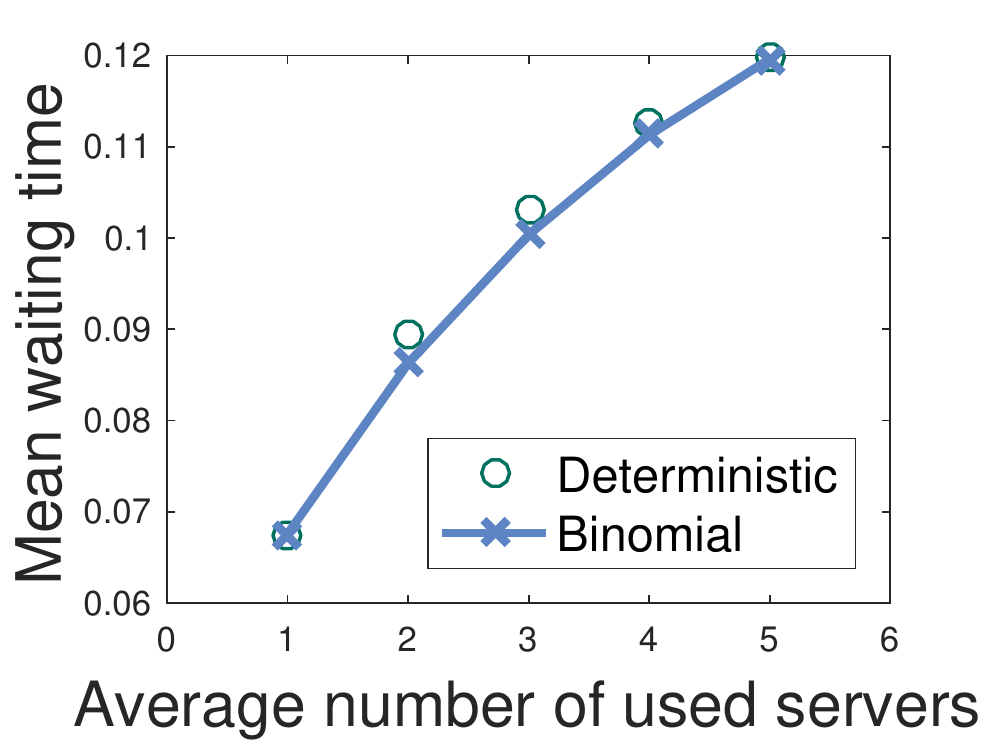}
		\caption{Det. vs. Stochastic Strategy}
		\label{fig:DetVsStoch}
	\end{subfigure}%
	~ 
	\begin{subfigure}[b]{0.23\textwidth}
		\includegraphics[width=\textwidth]{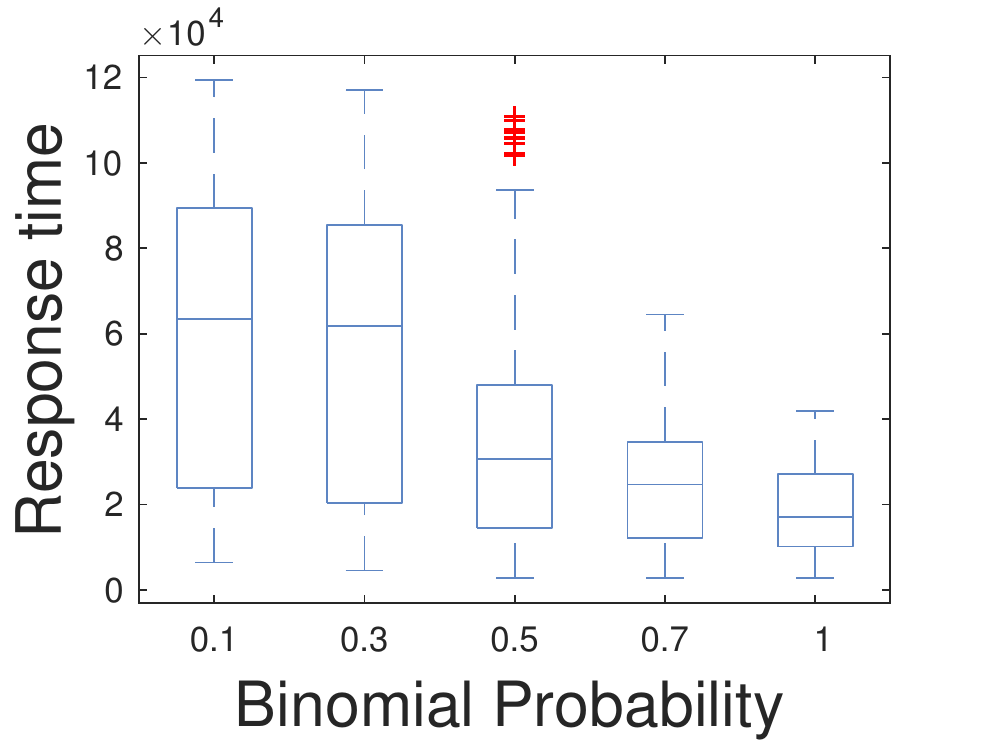}
		\caption{MPTCP as FJ Application}
		\label{fig:mptcp_eval}
	\end{subfigure}
	\caption{(Left) Deterministic vs. stochastic scheduling strategy for an application with specific $\varphi$ in a heterogeneous FJ system. Thm.~\ref{theorem:scaling_in_heterogeneous} shows that, in general, either strategy can be superior. (Right) The response time decreases with an increasing binomial probability, i.e., with increasing average number of Multipath TCP subflows.}
	\label{fig:LinScaling}
\end{figure}

In this section, we provide evaluations for two exemplary Fork-Join scenarios, namely,
\emph{(i)} a comparison of deterministic and stochastic scheduling strategies,
and \emph{(ii)} stochastic scheduling results for the transport protocol Multipath TCP. We consider partial as well as linear scaling benefit as given in \eqref{eq:power_scaling_assumption} and \eqref{eq:scaling_assumption}. 

\noindent \textbf{Evaluation of deterministic and stochastic strategies:}
In the following, we compare the average waiting times in a \emph{heterogeneous} FJ system that uses a binomial scheduling strategy with one using a corresponding deterministic strategy. Our aim is to show the benefit of Thm.~\ref{theorem:scaling_in_heterogeneous}. We consider renewal job arrivals with exponentially distributed inter-arrival times with parameter $\lambda=0.1$ at the ingress of an FJ system with $N=5$ servers each of which can be in a \emph{fast} or a \emph{slow} state with probability $0.5$. Hence, the service times are exponentially distributed with an average  of $\mu=1$ in the first state,  and $\mu=0.5$ in the second. We assume an application with a weak parallelization benefit $\varphi=0.2$. The rationale here is to let the system switch between a regime where the synchronization cost outweighs the scaling benefit, and another regime where the opposite holds true. Given a pool of $N$ available servers, Fig.~\ref{fig:DetVsStoch} compares the mean waiting time under a deterministic strategy that uses $1\leq S'\leq N$ servers to a stochastic strategy that uses an average number of servers $\Eof{S}=S'$. While this example shows that the stochastic strategy \emph{can} be superior to a comparable deterministic one, we know that in general the superiority of either strategy depends on the number of available servers $N$, the application specific parallelization benefit $\varphi$ and the utilization of the FJ system. This strengthens our arguments that for a known application that runs on a given FJ system Thm.~\ref{theorem:scaling_in_heterogeneous} provides the optimal scheduling strategy.

\begin{figure}
    \centering
    \includegraphics[width=\linewidth]{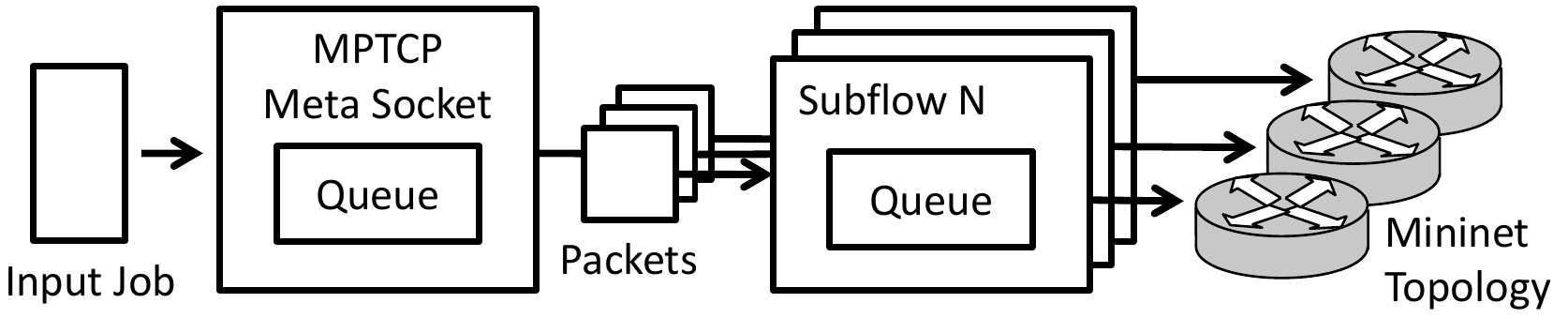}
    \caption{Network transfer evaluation setup: Multipath TCP splits jobs on multiple subflows.}
		\label{fig:mptcp_setup}
		\vspace{-3pt}
\end{figure}

\noindent \textbf{Number of Subflows with Multipath TCP:}
%
We evaluate the binomial strategy for a network data transfer scenario with linear scaling,
in which arriving jobs (\emph{datasets}) of varying sizes are transmitted.
Today's networks often provide several disjoint paths, e.g., for ECMP-based load balancing in data-center networks.
The scheduling strategy chooses the number of utilized network paths.
For a concrete evaluation,
we use the Multipath TCP (MPTCP) transport protocol as Fork-Join system~\cite{rfc6824_mptcp}.
Multipath TCP splits the data on multiple \emph{subflows} and joins them at the receiver side to ensure in-order data transfer for one logical TCP connection (see Fig.~\ref{fig:mptcp_setup}).


For the measurements, we use the MPTCP Linux kernel implementation~\cite{raiciu2012hard} and \emph{Mininet} to emulate topologies with disjoints paths. 
Fig.~\ref{fig:mptcp_eval} shows the response time given a binomial strategy, where the response time decreases with increasing binomial probability. Clearly, in this case the higher number of subflows overwhelms the synchronization penalty. Remarkably, we observe diminishing returns in terms of response time with increasing $p$, which directly translates to the average number of subflows.




%

\section{Related Work}
\label{sec:related-work}

In this section, we review related work on the 
Fork-Join queueing systems and their applications.
First inequalities for the stationary waiting time distribution in  \emph{GI/G/k}  queues are shown in \cite{kingman1970inequalities}. Martingale techniques have been used  in queueing theory, in particular, for providing exponential upper bounds by means of maximal inequalities in \cite{buffet1994exponential, Duffield1994} and later on in
\cite{poloczek2014scheduling}. The authors of \cite{poloczek2014scheduling} propose
a  characterization of queueing systems by bounding suitable martingale constructions, which allows embedding this queueing system characterization into the realm of stochastic network calculus.

An exact analysis of Fork-Join systems with more than two servers in a general setup remains elusive~\cite{baccelli1989fork,boxma1994queueing}, for it is  hard to find closed-form expressions for the steady-state distributions. 
Several works derive exact analytical results for special cases. The authors of \cite{Kim1989} obtain  transient and  steady-state solutions of the FJ queue in terms of  virtual waiting times.
%
The special case of an FJ system with two servers having exponential service times under Poissonian job arrivals 
is studied in \cite{flatto1984two}. Further, a multitude of useful approximations
\cite{ko2008sojourn,nelson1988approximate,kemper2012mean,varma1994interpolation}     and bounds  \cite{baccelli1989fork,balsamo1998bound,Rizk2015Sigmetrics,Rizk2016,Kesidis2015} are available in the literature.  In \cite{Xia2007}, the authors  study the scalability of a general FJ system with blocking, i.e., they study how the throughput of a general FJ system with blocking servers behaves as the number of nodes increases to infinity while the processing speed and buffer space of each node stay unchanged. Another interesting study of limiting behavior is done in \cite{Atar2012} where the authors study FJ networks with non-exchangeable tasks under a heavy traffic regime and show asymptotic equivalence between this network and its corresponding assembly network  with exchangeable tasks.
From the perspective of choosing task assignment policies in distributed server systems, the authors of \cite{harchol1999choosing} study various policies and suggest different optimal policies in different situations. Similarly, the work in \cite{Hyytia2012} seeks to quantify the benefits of splitting a task into different queues. 
It must be noted that the underlying premises in these works are quite dissimilar among themselves and from ours. 
We consider the works \cite{baccelli1989fork,Rizk2016} to be the closest to ours. 
While the basic instruments in deriving bounds in \cite{Rizk2016} are suitably constructed martingales, as they are in this work, the authors of \cite{Rizk2016} do not consider the notion of scheduling with respect to application specific scaling and only look at homogeneous servers. Further, \cite{baccelli1989fork} provides computable bounds for the  expected response times in FJ systems under renewal Poissonian arrival and exponential service times. Their methodology  differs from ours 
as they construct a 
tractable system to derive the bounds. We, on the contrary, concentrate on bounding the waiting time distributions and use these bounds to gain insights into application specific parallelization benefit and scheduling strategies therefrom.


Several contributions have been made to analyze the performance of applications that can be modeled by FJ systems such as MapReduce\cite{dean2008mapreduce,zaharia2008improving}. Performance optimization problems that arise for MapReduce systems are surveyed in \cite{polato2014comprehensive,Hashem2016}.
%
%
In \cite{pakize2014comprehensive}, the authors discuss different scheduling strategies regarding Hadoop MapReduce. 
Related work also considers 
the performance and pricing of EC2 instances such as 
on-demand instances (reliable, expensive) or spot instances (volatile, inexpensive). While there has been a number of articles studying spot pricing, mostly taking the provider's viewpoint such as \cite{zhang2014dynamic,jin2015towards}, the authors of \cite{zheng2015bid} take into account the user's standpoint too and explores bidding strategies analytically. These works can feed into our performance model as they essentially relate the obtained computing power, hence the service time distribution, to the monetary cost of computation. For instance, our model can thus be used to analyze a parallelized system where the number of utilized servers is modulated by the price curve of spot instances.


%
%
%
%
%
%
%
\section{Conclusions}

\label{sec:discussion}
In this paper we provide stochastic bounds on queueing performance metrics for heterogeneous Fork-Join systems under arbitrary level of parallelization benefit. Specifically, using a matching martingale construction we derive bounds on the waiting and response time distributions in this system. We model the application specific parallelization benefit in a given FJ system as a scaling parameter that affects the task service times and analytically show the impact of heterogeneity on this benefit.
We highlight a fundamental trade-off between the parallelization benefit and the FJ intrinsic synchronization penalty. Finally, we
propose optimal stochastic scheduling strategies in FJ systems for varying application specific parallelization benefits. We conclude our work with a simulation study that evaluates stochastic scheduling strategies in a Multipath TCP scenario while
optimizing the number of used paths to improve the system response time.
\section*{Acknowledgment}
This work has been funded by the German Research Foundation~(DFG) as part of projects C3 and B4 
within the Collaborative Research Center~(CRC) 1053 -- MAKI.

\bibliographystyle{IEEEtran}
\bibliography{main}

\begin{thebibliography}{10}
\providecommand{\url}[1]{#1}
\csname url@samestyle\endcsname
\providecommand{\newblock}{\relax}
\providecommand{\bibinfo}[2]{#2}
\providecommand{\BIBentrySTDinterwordspacing}{\spaceskip=0pt\relax}
\providecommand{\BIBentryALTinterwordstretchfactor}{4}
\providecommand{\BIBentryALTinterwordspacing}{\spaceskip=\fontdimen2\font plus
\BIBentryALTinterwordstretchfactor\fontdimen3\font minus
  \fontdimen4\font\relax}
\providecommand{\BIBforeignlanguage}[2]{{%
\expandafter\ifx\csname l@#1\endcsname\relax
\typeout{** WARNING: IEEEtranS.bst: No hyphenation pattern has been}%
\typeout{** loaded for the language `#1'. Using the pattern for}%
\typeout{** the default language instead.}%
\else
\language=\csname l@#1\endcsname
\fi
#2}}
\providecommand{\BIBdecl}{\relax}
\BIBdecl

\bibitem{EC2}
\BIBentryALTinterwordspacing
{Amazon Elastic Compute Cloud EC2}. [Online]. Available:
  \url{https://aws.amazon.com/ec2/}
\BIBentrySTDinterwordspacing

\bibitem{Spark}
\BIBentryALTinterwordspacing
{Apache Spark}. [Online]. Available: \url{http://spark.apache.org/}
\BIBentrySTDinterwordspacing

\bibitem{ash2014real}
R.~B. Ash, \emph{{Real Analysis and Probability}}.\hskip 1em plus 0.5em minus
  0.4em\relax Academic Press, 1972.

\bibitem{Atar2012}
R.~Atar, A.~Mandelbaum, and A.~Zviran, ``{Control of Fork-Join Networks in
  Heavy Traffic},'' in \emph{{Allerton}}, Oct 2012, pp. 823--830.

\bibitem{baccelli1989fork}
F.~Baccelli, A.~M. Makowski, and A.~Shwartz, ``{The Fork-Join Queue and Related
  Systems with Synchronization Constraints: Stochastic Ordering and Computable
  Bounds},'' \emph{{Advances in Applied Probability}}, pp. 629--660, 1989.

\bibitem{balsamo1998bound}
S.~Balsamo, L.~Donatiello, and N.~M.~V. Dijk, ``{Bound Performance Models of
  Heterogeneous Parallel Processing Systems},'' \emph{{IEEE Transactions on
  Parallel and Distributed Systems}}, vol.~9, no.~10, pp. 1041--1056, 1998.

\bibitem{boxma1994queueing}
O.~J. Boxma, G.~Koole, and Z.~Liu, \emph{{Queueing-theoretic Solution Methods
  for Models of Parallel and Distributed Systems}}.\hskip 1em plus 0.5em minus
  0.4em\relax {Centrum voor Wiskunde en Informatica, Department of Operations
  Research, Statistics, and System Theory}, 1994.

\bibitem{buffet1994exponential}
E.~Buffet and N.~Duffield, ``{Exponential Upper Bounds via Martingales for
  Multiplexers with Markovian Arrivals},'' \emph{{Journal of Applied
  Probability}}, pp. 1049--1060, 1994.

\bibitem{chong2007efficient}
J.~Chong, N.~Satish, B.~Catanzaro, K.~Ravindran, and K.~Keutzer, ``{Efficient
  Parallelization of H.264 Decoding with Macro Block Level Scheduling},'' in
  \emph{IEEE {ICME}}, 2007, pp. 1874--1877.

\bibitem{dean2008mapreduce}
J.~Dean and S.~Ghemawat, ``{MapReduce: Simplified Data Processing on Large
  Clusters},'' \emph{{Communications of the ACM}}, vol.~51, no.~1, pp.
  107--113, 2008.

\bibitem{Duffield1994}
N.~G. Duffield, ``{Exponential Bounds for Queues with Markovian Arrivals},''
  \emph{Queueing Systems}, vol.~17, no.~3, pp. 413--430, 1994.

\bibitem{flatto1984two}
L.~Flatto and S.~Hahn, ``{Two Parallel Queues Created by Arrivals with Two
  Demands I},'' \emph{{SIAM Journal on Applied Mathematics}}, vol.~44, no.~5,
  pp. 1041--1053, 1984.

\bibitem{rfc6824_mptcp}
\BIBentryALTinterwordspacing
A.~Ford, C.~Raiciu, M.~Handley, and O.~Bonaventure, ``{{TCP Extensions for
  Multipath Operation with Multiple Addresses}},'' RFC 6824, Internet
  Engineering Task Force. [Online]. Available:
  \url{http://www.ietf.org/rfc/rfc6824.txt}
\BIBentrySTDinterwordspacing

\bibitem{harchol1999choosing}
M.~Harchol-Balter, M.~E. Crovella, and C.~D. Murta, ``{On Choosing a Task
  Assignment Policy for a Distributed Server System},'' \emph{{Journal of
  Parallel and Distributed Computing}}, vol.~59, no.~2, pp. 204--228, 1999.

\bibitem{Hashem2016}
I.~A.~T. Hashem, N.~B. Anuar, A.~Gani, I.~Yaqoob, F.~Xia, and S.~U. Khan,
  ``{MapReduce: Review and Open Challenges},'' \emph{Scientometrics}, pp.
  1--34, 2016.

\bibitem{Hyytia2012}
E.~Hyyti{\"a} and S.~Aalto, ``{To Split or Not to Split: Selecting the Right
  Server with Batch Arrivals},'' \emph{Operations Research Letters}, vol.~41,
  no.~4, pp. 325 -- 330, 2013.

\bibitem{jin2015towards}
H.~Jin, X.~Wang, S.~Wu, S.~Di, and X.~Shi, ``{Towards Optimized Fine-grained
  Pricing of IAAS Cloud Platform},'' \emph{{IEEE Transactions on Cloud
  Computing}}, vol.~3, no.~4, pp. 436--448, 2015.

\bibitem{kemper2012mean}
B.~Kemper and M.~Mandjes, ``{Mean Sojourn Times in Two-queue Fork-Join Systems:
  Bounds and Approximations},'' \emph{{OR Spectrum}}, vol.~34, no.~3, pp.
  723--742, 2012.

\bibitem{Kesidis2015}
G.~Kesidis, Y.~Shan, B.~Urgaonkar, and J.~Liebeherr, ``{Network Calculus for
  Parallel Processing},'' \emph{{SIGMETRICS Perform. Eval. Rev.}}, vol.~43,
  no.~2, pp. 48--50, Sep. 2015.

\bibitem{KhudaBukhsh2016techreport}
\BIBentryALTinterwordspacing
W.~R. KhudaBukhsh, A.~Rizk, A.~Fr\"ommgen, and H.~Koeppl, ``{Optimizing
  Stochastic Scheduling in Fork-Join Queueing Models: Bounds and
  Applications},'' 2016, {Extended version}. [Online]. Available:
  \url{https://arxiv.org/abs/1612.05486}
\BIBentrySTDinterwordspacing

\bibitem{Kim1989}
C.~Kim and A.~K. Agrawala, ``{Analysis of the Fork-join Queue},'' \emph{{IEEE
  Transactions on Computers}}, vol.~38, no.~2, pp. 250--255, Feb 1989.

\bibitem{kingman1970inequalities}
J.~Kingman, ``{Inequalities in the Theory of Queues},'' \emph{{Journal of the
  Royal Statistical Society. Series B (Methodological)}}, pp. 102--110, 1970.

\bibitem{ko2008sojourn}
S.-S. Ko and R.~F. Serfozo, ``{Sojourn Times in G/M/1 Fork-join Networks},''
  \emph{{Naval Research Logistics (NRL)}}, vol.~55, no.~5, pp. 432--443, 2008.

\bibitem{mesa2009scalability}
M.~A. Mesa, A.~Ram{\'\i}rez, A.~Azevedo, C.~Meenderinck, B.~Juurlink, and
  M.~Valero, ``{Scalability of Macroblock-level Parallelism for H.264
  Decoding},'' in \emph{{ICPADS}}, 2009, pp. 236--243.

\bibitem{nelson1988approximate}
R.~Nelson and A.~N. Tantawi, ``{Approximate Analysis of Fork/join
  Synchronization in Parallel Queues},'' \emph{{IEEE Transactions on
  Computers}}, vol.~37, no.~6, pp. 739--743, Jun 1988.

\bibitem{pakize2014comprehensive}
S.~R. Pakize, ``{A Comprehensive View of Hadoop MapReduce Scheduling
  Algorithms},'' \emph{{International Journal of Computer Networks \&
  Communications Security}}, vol.~2, no.~9, pp. 308--317, 2014.

\bibitem{polato2014comprehensive}
I.~Polato, R.~R{\'e}, A.~Goldman, and F.~Kon, ``{A Comprehensive View of Hadoop
  Research-A Systematic Literature Review},'' \emph{{Journal of Network and
  Computer Applications}}, vol.~46, pp. 1--25, 2014.

\bibitem{poloczek2014scheduling}
F.~Poloczek and F.~Ciucu, ``{Scheduling Analysis with Martingales},''
  \emph{{Performance Evaluation}}, vol.~79, pp. 56--72, 2014.

\bibitem{raiciu2012hard}
C.~Raiciu, C.~Paasch, S.~Barre, A.~Ford, M.~Honda, F.~Duchene, O.~Bonaventure,
  and M.~Handley, ``{How Hard Can It Be? Designing and Implementing a
  Deployable Multipath TCP},'' in \emph{USENIX NSDI}, 2012, pp. 29--29.

\bibitem{Rizk2015Sigmetrics}
A.~Rizk, F.~Poloczek, and F.~Ciucu, ``{Computable Bounds in Fork-Join Queueing
  Systems},'' \emph{{SIGMETRICS Perform. Eval. Rev.}}, vol.~43, no.~1, pp.
  335--346, Jun. 2015.

\bibitem{Rizk2016}
------, ``{Stochastic bounds in Fork-Join queueing systems under full and
  partial mapping},'' \emph{Queueing Systems}, vol.~83, no.~3, pp. 261--291,
  2016.

\bibitem{SubramanyaGSIS15}
S.~Subramanya, T.~Guo, P.~Sharma, D.~E. Irwin, and P.~J. Shenoy, ``{SpotOn: A
  Batch Computing Service for the Spot Market},'' in \emph{SoCC}, 2015, pp.
  329--341.

\bibitem{varma1994interpolation}
S.~Varma and A.~M. Makowski, ``{Interpolation Approximations for Symmetric
  Fork-Join Queues},'' \emph{{Performance Evaluation}}, vol.~20, no. 1-3, pp.
  245--265, 1994.

\bibitem{Xia2007}
C.~H. Xia, Z.~Liu, D.~Towsley, and M.~Lelarge, ``Scalability of fork/join
  queueing networks with blocking,'' \emph{SIGMETRICS Perform. Eval. Rev.},
  vol.~35, no.~1, pp. 133--144, Jun. 2007.

\bibitem{zaharia2008improving}
M.~Zaharia, A.~Konwinski, A.~D. Joseph, R.~H. Katz, and I.~Stoica, ``{Improving
  MapReduce Performance in Heterogeneous Environments},'' in \emph{{USENIX
  OSDI}}, vol.~8, no.~4, 2008, p.~7.

\bibitem{zhang2014dynamic}
L.~Zhang, Z.~Li, and C.~Wu, ``{Dynamic Resource Provisioning in Cloud
  Computing: A Randomized Auction Approach},'' in \emph{{IEEE INFOCOM}}, 2014,
  pp. 433--441.

\bibitem{zheng2015bid}
L.~Zheng, C.~Joe-Wong, C.~W. Tan, M.~Chiang, and X.~Wang, ``{How to Bid the
  Cloud},'' \emph{{ACM SIGCOMM Comput. Commun. Rev.}}, vol.~45, no.~4, pp.
  71--84, 2015.

\end{thebibliography}



\ifthenelse{\boolean{longVersion}}{\ \newpage \newpage
\section{Appendix}
\label{sec:Appendix}

Before we present our proofs, let us make a remark that will be useful throughout the discourse.
\begin{myRemark}
	If $\{ X_k, \mathcal{A}_k\} $ and $\{ Y_k, \mathcal{A}_k\}$ are submartingales, then so is $\{ \max(X_k,Y_k) , \mathcal{A}_k\}$ (see \emph{7.3.2.(e) Comments} of \cite{ash2014real} for a proof). Treating martingales as submartingales and extending the above mentioned result to accomodate maximum over a  finite collection of submartinagles, we establish that $\{ X(k), \mathcal{F}_k \}_{k \in \setOfNonnegativeIntegers} $ is a submartingale whenever $\{ X_n(k), \mathcal{A}(k)\}_{k \in \setOfNonnegativeIntegers} $ is a submartingale (or a martingale) for each $ n \in \setN{N}$ and $ X(k) \defeq \max_{n \in \setN{N}} X_n(k) , \forall k \in \setOfNonnegativeIntegers$.
	\label[remark]{remark:max_submartingale}
\end{myRemark}


\section*{}

\begin{proof}[Proof of  \Cref{theorem:RenewalsNon-blocking}\\ ] 
	First notice that the stability condition given in 
	 $ \max_{n \in \setN{N}} \Eof{X_{n,1}} < \Eof{T_1} $ guarantees the existence of $\theta_n >0$ such that $ \alpha_n( \theta_n) \beta(\theta_n) =1 $ for all $n \in \setN{N}$ (see \cite{poloczek2014scheduling, boxma1994queueing}). Hence, $\tilde{\theta} > 0$ is well defined.

	Consider the filtration 
	\begin{equation*}
	\mathcal{F}_k \defeq \sigma( \{ X_{n,i} \}_{n \in \setN{N}, i \leq k}, \{ T_i  \}_{i \leq k} ) \eqkomma
	\end{equation*}
	for all $k \in \setOfNonnegativeIntegers $, where $\sigma( \{ X_{n,i} \}_{n \in \setN{N}, i \leq k}, \{ T_i  \}_{i \leq k} ) $ denotes the smallest $\sigma$-field generated 
	by $ \{ X_{n,i} \}_{n \in \setN{N}, i \leq k}, \{ T_i  \}_{i \leq k}  $. 
	
		\paragraph*{Bounding the waiting time}
	For each $n \in \setN{N}$, define the stochastic process
	\begin{equation*}
	Z_n(k) \defeq \myExp{ \theta_n \sum_{i=1}^{k} (X_{n,i} - T_i)}   , \forall k \in \setOfNonnegativeIntegers  \eqpunkt 
	\end{equation*}
	We shall show that $\{ Z_n(k), \mathcal{F}_k \}_{k \in \setOfNonnegativeIntegers} $ is a martingale. See that, 
	
	\begin{align*}
		\Eof{Z_n(k) \mid \mathcal{F}_{k-1} } ={} &  \Eof{  \myExp{ \theta_n \sum_{i=1}^{k} (X_{n,i} - T_i)}  \mid  \mathcal{F}_{k-1} } \\
		={}&  \myExp{ \theta_n \sum_{i=1}^{k-1} (X_{n,i} - T_i)} \\
		& \times \Eof{  \myExp{ \theta_n  (X_{n,k} - T_k)}   \mid  \mathcal{F}_{k-1}  } \\
		={}& Z_n(k-1) \Eof{  \myExp{ \theta_n  (X_{n,k} - T_k)}  } \\
		={}& Z_n(k-1)  \alpha_n( \theta_n) \beta(\theta_n)  \\
		={} & Z_n(k-1) \eqkomma
	\end{align*}
	exploiting our independence assumptions throughout. The above implies that  $\{ Z_n(k), \mathcal{F}_k \}_{k \in \setOfNonnegativeIntegers} $ is a martingale. 

By virtue of the  sub- and supermartingale inequalities due to Doob (see \emph{Problems 3 (c)} in chapter 7 of \cite{ash2014real}), we have 
	
	\begin{eqnarray}
		\probOf{ \sup_{k \in \setOfNonnegativeIntegers} Z_n(k) \geq  \sigma } \leq \frac{ \sup_{k \in \setOfNonnegativeIntegers} \Eof{ Z_n^{+} (k)}  }{\sigma} \eqkomma
		\label{eq:Doob_inequality}
	\end{eqnarray}
	treating our martingale $Z_n(k)$ as a submartingale, for $\sigma \geq 0$ and for each $n \in \setN{N}$. Now, our martingales are so constructed that 
	
	\begin{eqnarray*}
		Z_n^{+} (k) & \defeq & Z_n (k)  \wedge 0 = Z_n(k) \forall k \in \setOfNonnegativeIntegers  \\
		\implies \Eof{ Z_n^{+} (k)} &=& \Eof{ Z_n(k)} =1 \forall k \in \setOfNonnegativeIntegers  \\
		\implies  \sup_{k \in \setOfNonnegativeIntegers} \Eof{ Z_n^{+} (k)}  &=& 1 \eqpunkt
	\end{eqnarray*}
	Therefore, we have 
	\begin{eqnarray}
		\probOf{ \sup_{k \in \setOfNonnegativeIntegers} Z_n(k) \geq  \sigma } \leq  \frac{ 1 }{\sigma} \eqpunkt
		\label[equation]{eq:Doob_bound2}
	\end{eqnarray} 
	Now define $\tilde{\theta} \defeq \min_{n \in \setN{N}} \theta_n $.
	
	Finally we bound the tail probabilities of the waiting time $W$ as follows 
	\begin{eqnarray*}
		\probOf{W \geq \sigma } & =& \probOf{ \sup_{n \in \setN{N} } \{   \sup_{k \in \setOfNonnegativeIntegers}  \{ \sum_{i=1}^{k} (X_{n,i} - T_i) \}  \}   \geq \sigma} \\
		&=& \probOf{  \cup_{  n \in \setN{N}  } \{   \sup_{k \in \setOfNonnegativeIntegers}  \{ \sum_{i=1}^{k} (X_{n,i} - T_i) \}  \}   \geq \sigma  } \\
		& \leq &  \sum_{n \in \setN{N} }  \probOf{   \sup_{k \in \setOfNonnegativeIntegers}  \{ \sum_{i=1}^{k} (X_{n,i} - T_i) \}    \geq   \sigma} \\
		&=& \sum_{n \in \setN{N} }  \probOf{   \sup_{k \in \setOfNonnegativeIntegers} Z_n(k)   \geq   \myExp{  \theta_n \sigma } } \\
		&=& \sum_{n \in \setN{N} } \myExp{ - \theta_n \sigma }   \\
		&=&   \myExp{ - \tilde{\theta}  \sigma  }  \sum_{n \in \setN{N} } \myExp{ -  ( \theta_n  -  \tilde{\theta}   )\sigma }    \eqkomma \\
	\end{eqnarray*}
	by  \emph{Boole's inequality} and~\eqref{eq:Doob_bound2}. For further simplifications, we shall also use the following bound 
	\begin{align*}
		\probOf{W \geq \sigma } & \leq N  \myExp{ - \tilde{\theta}  \sigma  }   \eqpunkt 
	\end{align*}

	\paragraph*{Bounding the response time}
	
	Define the stochastic process, for each $n \in \setN{N}$, 
	\begin{equation*}
	Y_n(k) \defeq  \myExp{ \theta_n ( \sum_{i=0}^{k} X_{n,i} -  \sum_{i=1}^{k}  T_i ) } , \forall k \in \setOfNonnegativeIntegers \eqpunkt 
	\end{equation*}
	
	We shall show that $\{ Y_n(k), \mathcal{F}_k \}_{k \in \setOfNonnegativeIntegers} $ is a martingale. See that, 
	\begin{align*}
		\Eof{Y_n(k) \mid \mathcal{F}_{k-1} } = {}&  \Eof{ \myExp{ \theta_n ( \sum_{i=0}^{k} X_{n,i} -  \sum_{i=1}^{k}  T_i ) } \mid  \mathcal{F}_{k-1} } \\
		={}& \myExp{ \theta_n ( \sum_{i=0}^{k-1} X_{n,i} -  \sum_{i=1}^{k-1}  T_i ) } \\
		& \times  \Eof{  \myExp{ \theta_n  (X_{n,k} - T_k)}   \mid  \mathcal{F}_{k-1}  } \\
		={}& Y_n(k-1) \Eof{  \myExp{ \theta_n  (X_{n,k} - T_k)}  } \\
		={}& Y_n(k-1)  \alpha_n( \theta_n) \beta(\theta_n)  \\
		={}& Y_n(k-1) \eqkomma
	\end{align*}
	exploiting our independence assumptions throughout. The above implies that  $\{ Y_n(k), \mathcal{F}_k \}_{k \in \setOfNonnegativeIntegers} $ is a martingale.

	By virtue of the  sub- and supermartingale inequalities due to Doob (see \emph{Problems 3 (c)} in chapter 7 of \cite{ash2014real}), we have 
	
	\begin{eqnarray}
	\probOf{ \sup_{k \in \setOfNonnegativeIntegers} Y_n(k) \geq  \sigma } \leq \frac{ \sup_{k \in \setOfNonnegativeIntegers} \Eof{ Y_n^{+} (k)}  }{\sigma} \eqkomma
	\end{eqnarray}
	treating our martingale $Y_n(k)$ as a submartingale, for $\sigma \geq 0$ and for each $n \in \setN{N}$. Now, our martingales are so constructed that 
	
	\begin{eqnarray*}
		Y_n^{+} (k) & \defeq & Y_n (k)  \wedge 0 = Y_n(k)  \;  \forall k \in \setOfNonnegativeIntegers  \\
		\implies \Eof{ Y_n^{+} (k)} &=& \Eof{ Y_n(k)} \\
		& = &   \Eof{\myExp{\theta_n X_{n,0}}} \prod_{i=1}^{k} \alpha_n( \theta_n) \beta(\theta_n)  \\
		& =&     \alpha_n( \theta_n)    \; 
	\forall k \in \setOfNonnegativeIntegers  \\
		\implies  \sup_{k \in \setOfNonnegativeIntegers} \Eof{ Y_n^{+} (k)} &=  &   \alpha_n( \theta_n)    \; 
		\forall k \in \setOfNonnegativeIntegers  
    \eqpunkt
	\end{eqnarray*}
	Therefore, we have 
	\begin{eqnarray}
	\probOf{ \sup_{k \in \setOfNonnegativeIntegers} Y_n(k) \geq  \sigma } \leq      \frac{   \alpha_n( \theta_n)  }{\sigma} \eqpunkt
	\label[equation]{eq:Doob_bound3}
	\end{eqnarray} 
	Now define $\tilde{\theta} \defeq \min_{n \in \setN{N}} \theta_n $.
	
	Finally we bound the tail probabilities of the response time $R$ as follows 
	\begin{eqnarray*}
		\probOf{R  \geq \sigma } & =& \probOf{ \sup_{n \in \setN{N} } \{   \sup_{k \in \setOfNonnegativeIntegers}  \{ \sum_{i=0}^{k} X_{n,i} -  \sum_{i=1}^{k}  T_i  \}  \}   \geq \sigma} \\
		&=& \probOf{  \cup_{  n \in \setN{N}  } \{   \sup_{k \in \setOfNonnegativeIntegers}  \{ \sum_{i=0}^{k} X_{n,i} -  \sum_{i=1}^{k}  T_i  \}  \}   \geq \sigma  } \\
		& \leq &  \sum_{n \in \setN{N} }  \probOf{   \sup_{k \in \setOfNonnegativeIntegers}  \{ \sum_{i=0}^{k} X_{n,i} -  \sum_{i=1}^{k}  T_i  \}    \geq   \sigma} \\
		&=& \sum_{n \in \setN{N} }  \probOf{   \sup_{k \in \setOfNonnegativeIntegers} Y_n(k)   \geq   \myExp{  \theta_n \sigma } } \\
		&=& \sum_{n \in \setN{N} }   \alpha_n( \theta_n) \myExp{ - \theta_n \sigma }   \\
		&=&   \myExp{ - \tilde{\theta}  \sigma  }  \sum_{n \in \setN{N} }  \alpha_n( \theta_n)  \myExp{ -  ( \theta_n  -  \tilde{\theta}   )\sigma }    \eqkomma \\
	\end{eqnarray*}
	by  \emph{Boole's inequality} and~\eqref{eq:Doob_bound3}. For further simplifications, we shall also use the following bound 
	\begin{align*}
	\probOf{ R \geq \sigma } & \leq  \myExp{ - \tilde{\theta}  \sigma  }   \sum_{n \in \setN{N} }  \alpha_n( \theta_n)  \eqpunkt 
	\end{align*}

\end{proof}

\begin{proof}[Proof of \Cref{theorem:PRMappingRenewalsNon-blocking} \\ ]

	Consider the filtration 
	\begin{equation*}
	\mathcal{F}_k \defeq \sigma( \{ \tilde{X}_{n,i} \}_{n \in \setN{N}, i \leq k}, \{ T_i  \}_{i \leq k} ) \eqkomma
	\end{equation*}
	for all $k \in \setOfNonnegativeIntegers $, where $\sigma( \{ \tilde{X}_{n,i} \}_{n \in \setN{N}, i \leq k}, \{ T_i  \}_{i \leq k} ) $ denotes the smallest $\sigma$-field generated 
	by $ \{ \tilde{X}_{n,i} \}_{n \in \setN{N}, i \leq k}, \{ T_i  \}_{i \leq k}  $. 
	
		\paragraph*{Bounding the waiting time}
	
For each $n \in \setN{N}$, define the stochastic process
	\begin{equation*}
	Z_n(k) \defeq \myExp{ \theta_n \sum_{i=1}^{k} (\tilde{X}_{n,i} - T_i)}   , \forall k \in \setOfNonnegativeIntegers  \eqpunkt 
	\end{equation*}
	We shall show that $\{ Z_n(k), \mathcal{F}_k \}_{k \in \setOfNonnegativeIntegers} $ is a martingale. See that, 
	
	\begin{align*}
	&	\Eof{Z_n(k) \mid \mathcal{F}_{k-1} } \\
	={} &  \Eof{  \myExp{ \theta_n \sum_{i=1}^{k} ( \tilde{X}_{n,i} - T_i)}  \mid  \mathcal{F}_{k-1} } \\
		={}&  \myExp{ \theta_n \sum_{i=1}^{k-1} (\tilde{X}_{n,i} - T_i)} \Eof{  \myExp{ \theta_n  (\tilde{X}_{n,k} - T_k)}   \mid  \mathcal{F}_{k-1}  } \\
		={}& Z_n(k-1) \Eof{  \myExp{ \theta_n  (\tilde{X}_{n,k} - T_k)}  } \\
		={}& Z_n(k-1)  \alpha_n^{*}( \theta_n) \beta(\theta_n)  \\
		={}& Z_n(k-1) \eqkomma
	\end{align*}
	exploiting our independence assumptions throughout. The above implies that  $\{ Z_n(k), \mathcal{F}_k \}_{k \in \setOfNonnegativeIntegers} $ is a martingale.

	By virtue of the  sub- and supermartingale inequalities due to Doob (see \emph{Problems 3 (c)} in chapter 7 of \cite{ash2014real}), we have 
	
	\begin{eqnarray}
	\probOf{ \sup_{k \in \setOfNonnegativeIntegers} Z_n(k) \geq  \sigma } \leq \frac{ \sup_{k \in \setOfNonnegativeIntegers} \Eof{ Z_n^{+} (k)}  }{\sigma} \eqkomma
	\end{eqnarray}
	treating our martingale $Z_n(k)$ as a submartingale, for $\sigma \geq 0$ and for each $n \in \setN{N}$. Now, our martingales are so constructed that 
	
	\begin{eqnarray*}
		Z_n^{+} (k) & \defeq & Z_n (k)  \wedge 0 = Z_n(k) \forall k \in \setOfNonnegativeIntegers  \\
		\implies \Eof{ Z_n^{+} (k)} &=& \Eof{ Z_n(k)} =1 \forall k \in \setOfNonnegativeIntegers  \\
		\implies  \sup_{k \in \setOfNonnegativeIntegers} \Eof{ Z_n^{+} (k)}  &=& 1 \eqpunkt
	\end{eqnarray*}
	Therefore, we have 
	\begin{eqnarray}
	\probOf{ \sup_{k \in \setOfNonnegativeIntegers} Z_n(k) \geq  \sigma } \leq  \frac{ 1 }{\sigma} \eqpunkt
	\label[equation]{eq:Doob_bound4}
	\end{eqnarray} 
	Now define $\tilde{\theta} \defeq \min_{n \in \setN{N}} \theta_n $.
	
	Finally we bound the tail probabilities of the waiting time $W$ as follows 
	\begin{eqnarray*}
		\probOf{W \geq \sigma } & =& \probOf{ \sup_{n \in \setN{N} } \{   \sup_{k \in \setOfNonnegativeIntegers}  \{ \sum_{i=1}^{k} (X_{n,i} - T_i) \}  \}   \geq \sigma} \\
		&=& \probOf{  \cup_{  n \in \setN{N}  } \{   \sup_{k \in \setOfNonnegativeIntegers}  \{ \sum_{i=1}^{k} (X_{n,i} - T_i) \}  \}   \geq \sigma  } \\
		& \leq &  \sum_{n \in \setN{N} }  \probOf{   \sup_{k \in \setOfNonnegativeIntegers}  \{ \sum_{i=1}^{k} (X_{n,i} - T_i) \}    \geq   \sigma} \\
		&=& \sum_{n \in \setN{N} }  \probOf{   \sup_{k \in \setOfNonnegativeIntegers} Z_n(k)   \geq   \myExp{  \theta_n \sigma } } \\
		&=& \sum_{n \in \setN{N} } \myExp{ - \theta_n \sigma }   \\
		&=&   \myExp{ - \tilde{\theta}  \sigma  }  \sum_{n \in \setN{N} } \myExp{ -  ( \theta_n  -  \tilde{\theta}   )\sigma }    \eqkomma \\
	\end{eqnarray*}
	by  \emph{Boole's inequality} and~\eqref{eq:Doob_bound4}. For further simplifications, we shall also use the following bound 
	\begin{align*}
	\probOf{W \geq \sigma } & \leq N  \myExp{ - \tilde{\theta}  \sigma  }   \eqpunkt 
	\end{align*}

	\paragraph*{Bounding the response time}
	
	Define the stochastic process, for each $n \in \setN{N}$, 
	\begin{equation*}
	Y_n(k) \defeq  \myExp{ \theta_n ( X_{n,0}+  \sum_{i=1}^{k} \tilde{X}_{n,i} -  \sum_{i=1}^{k}  T_i ) } , \forall k \in \setOfNonnegativeIntegers \eqpunkt 
	\end{equation*}
	
	We shall show that $\{ Y_n(k), \mathcal{F}_k \}_{k \in \setOfNonnegativeIntegers} $ is a martingale. See that, 
	\begin{align*}
	&	\Eof{Y_n(k) \mid \mathcal{F}_{k-1} } \\
	={} &  \Eof{ \myExp{ \theta_n (X_{n,0} + \sum_{i=1}^{k} \tilde{X}_{n,i} -  \sum_{i=1}^{k}  T_i ) } \mid  \mathcal{F}_{k-1} } \\
		={} & \myExp{ \theta_n (X_{n,0} + \sum_{i=1}^{k-1} \tilde{X}_{n,i} -  \sum_{i=1}^{k-1}  T_i ) } \\
		& \times
		 \Eof{  \myExp{ \theta_n  (\tilde{X}_{n,k} - T_k)}   \mid  \mathcal{F}_{k-1}  } \\
		={} & Y_n(k-1) \Eof{  \myExp{ \theta_n  (\tilde{X}_{n,k} - T_k)}  } \\
		={} & Y_n(k-1)  \alpha_n^{*}( \theta_n) \beta(\theta_n)  \\
		={} & Y_n(k-1) \eqkomma
	\end{align*}
	exploiting our independence assumptions throughout. The above implies that  $\{ Y_n(k), \mathcal{F}_k \}_{k \in \setOfNonnegativeIntegers} $ is a martingale.

	By virtue of the  sub- and supermartingale inequalities due to Doob (see \emph{Problems 3 (c)} in chapter 7 of \cite{ash2014real}), we have 
	
	\begin{eqnarray}
	\probOf{ \sup_{k \in \setOfNonnegativeIntegers} Y_n(k) \geq  \sigma } \leq \frac{ \sup_{k \in \setOfNonnegativeIntegers} \Eof{ Y_n^{+} (k)}  }{\sigma} \eqkomma
	\end{eqnarray}
	treating our martingale $Y_n(k)$ as a submartingale, for $\sigma \geq 0$ and for each $n \in \setN{N}$. Now, our martingales are so constructed that 
	
	\begin{eqnarray*}
		Y_n^{+} (k) & \defeq & Y_n (k)  \wedge 0 = Y_n(k)  \;  \forall k \in \setOfNonnegativeIntegers  \\
		\implies \Eof{ Y_n^{+} (k)} &=& \Eof{ Y_n(k)} \\
		& = &   \Eof{\myExp{\theta_n X_{n,0}}} \prod_{i=1}^{k} \alpha_n^*( \theta_n) \beta(\theta_n)  \\
		& =&     \alpha_n( \theta_n)    \; 
		\forall k \in \setOfNonnegativeIntegers  \\
		\implies  \sup_{k \in \setOfNonnegativeIntegers} \Eof{ Y_n^{+} (k)} &=  &   \alpha_n( \theta_n)    \; 
		\forall k \in \setOfNonnegativeIntegers  
		\eqpunkt
	\end{eqnarray*}
	Therefore, we have 
	\begin{eqnarray}
	\probOf{ \sup_{k \in \setOfNonnegativeIntegers} Y_n(k) \geq  \sigma } \leq      \frac{   \alpha_n( \theta_n)  }{\sigma} \eqpunkt
	\label[equation]{eq:Doob_bound5}
	\end{eqnarray} 
	Now define $\tilde{\theta} \defeq \min_{n \in \setN{N}} \theta_n $.
	
	Finally we bound the tail probabilities of the response time $R$ as follows 
	\begin{eqnarray*}
		\probOf{R  \geq \sigma } & =& \probOf{ \sup_{n \in \setN{N} } \{   \sup_{k \in \setOfNonnegativeIntegers}  \{   X_{n,0} +  \sum_{i=1}^{k} \tilde{X}_{n,i} -  \sum_{i=1}^{k}  T_i  \}  \}   \geq \sigma} \\
		&=& \probOf{  \cup_{  n \in \setN{N}  } \{   \sup_{k \in \setOfNonnegativeIntegers}  \{ X_{n,0} +  \sum_{i=1}^{k} \tilde{X}_{n,i} -  \sum_{i=1}^{k}  T_i  \}  \}   \geq \sigma  } \\
		& \leq &  \sum_{n \in \setN{N} }  \probOf{   \sup_{k \in \setOfNonnegativeIntegers}  \{ X_{n,0} +  \sum_{i=1}^{k} \tilde{X}_{n,i} -  \sum_{i=1}^{k}  T_i  \}    \geq   \sigma} \\
		&=& \sum_{n \in \setN{N} }  \probOf{   \sup_{k \in \setOfNonnegativeIntegers} Y_n(k)   \geq   \myExp{  \theta_n \sigma } } \\
		&=& \sum_{n \in \setN{N} }   \alpha_n( \theta_n) \myExp{ - \theta_n \sigma }   \\
		&=&   \myExp{ - \tilde{\theta}  \sigma  }  \sum_{n \in \setN{N} }  \alpha_n( \theta_n)  \myExp{ -  ( \theta_n  -  \tilde{\theta}   )\sigma }    \eqkomma \\
	\end{eqnarray*}
	by  \emph{Boole's inequality} and~\eqref{eq:Doob_bound5}. For further simplifications, we shall also use the following bound 
	\begin{align*}
	\probOf{ R \geq \sigma } & \leq  \myExp{ - \tilde{\theta}  \sigma  }   \sum_{n \in \setN{N} }  \alpha_n( \theta_n)  \eqpunkt 
	\end{align*}

\end{proof}

\begin{myLemma}
	\begin{enumerate}
		\item	Suppose $X \sim \Bin{N}{p}$. Then, the following can be obtained, for $a>0$
		\begin{eqnarray*}
			\Eof{X e^{-aX}} & =&    \frac{N p e^ {-a} }{1- q^N}   (p e^{-a} + q  ) ^{N-1}   \\
			\Eof{X^2 e^{-aX}} & =&    \frac{N p e^ {-a} }{1- q^N}  (  Np e^{-a} +q  ) (    p e^{-a} + q    )^{N-2}  \eqpunkt
		\end{eqnarray*}
		
		\item If~$X$ is distributed uniformly over~$\setN{N}$, then for $a>0$ the following holds
		\begin{eqnarray*}
			\Eof{X e^{-aX}} & =&   \frac{e^{-a}}{N (1- e^{-a})}  [ \frac{1-  e^{ - (N+1)a } } {(1- e^{-a}) }\\
			&&  - (N+1) e^{- a N} ]    \\
			\Eof{X^2 e^{-aX}} & =&     \frac{ e ^ {-2 a}}{N (1- e^{-a})  }  [    2 \frac{(1- e^{-(N+1)a} ) }{   (1- e^{-a}) ^2 }    \\
			&& - \frac{  2(N+1)e^{-Na}  - (1-  e^{ - (N+1)a })  }{ (1- e^{-a})  }  \\
			&& - (N+1)  (  N e^{- (N-1)a} +  e^{- a N}  ) ]  \eqpunkt
		\end{eqnarray*}
	\end{enumerate}
	\label[lemma]{lemma:binomial-uniform}
\end{myLemma}

\begin{proof}[Proof of \Cref{lemma:binomial-uniform}: ]
	\begin{enumerate}
		\item 	First note that

		\begin{align*}
		&	\Eof{X e^{- a X}} \\
		={}& \sum_{s \in \setN{N}} s  e^{- a s}  \probOf{ X=s } \\
		={}& \sum_{s \in \setN{N}} s  e^{- a s}   \frac{\binom{N}{s} p^s q^ {N-s}}{1 - q^N}  \\
		={}& \sum_{s \in \setN{N}} s  e^{- a s} \frac{N}{s}  \frac{\binom{N-1}{s-1} p^s q^ {N-s}}{1 - q^N}  \\
		={}&  \frac{N p e^{- a} }{1 - q^N}  \sum_{s \in \setN{N}}  \binom{N-1}{s-1} { (p e^{-a} ) }^{s-1} q^ {(N-1)-(s-1)} \\
		={} & \frac{N p e^{- a} }{1 - q^N} ( p e^{- a} +q  )^{N-1}  \eqpunkt 
		\end{align*}
		
		Now, see that 
		
		\begin{align*}
		\Eof{ X^2 e^{- a X}}  = {}& 	\Eof{ X(X-1) e^{- a X} + Xe^{- a X}} \\
		={} & 	\Eof{ X(X-1) e^{- a X} } +  \Eof{ X e^{- aX}}   \eqpunkt 
		\end{align*}
		Now, 
		\begin{align*}
		&	\Eof{ X(X-1) e^{- a X  }} \\
		={}& \sum_{s \in \setN{N}} s(s-1)  e^{-a s}  \probOf{ X=s } \\
		={}& \sum_{s \in \setN{N}} s(s-1)  e^{- a s}   \frac{\binom{N}{s} p^s q^ {N-s}}{1 - q^N}  \\
		={}& \sum_{s \in \setN{N}  \setminus \{1\} } s(s-1)  e^{-a s} \frac{N(N-1)}{s(s-1)}  \frac{\binom{N-2}{s-2} p^s q^ {N-s}}{1 - q^N}  \\
		={}&  \frac{N(N-1)    (p e^{- a} )^2 }{1 - q^N}  \sum_{s \in \setN{N}   \setminus \{1\} }  \binom{N-2}{s-2} { (p e^{- a } ) }^{s-2} q^ {(N-2)-(s-2)} \\
		={}& \frac{N(N-1)    (p e^{- a} )^2 }{1 - q^N}   ( p e^{- a} +q  )^{N-2}  \eqpunkt 
		\end{align*}
		Therefore, we get 
		\begin{align*}
		& 	\Eof{ X^2 e^{- a X }} \\
		={}& 	\Eof{ X(X-1) e^{- a X } } +  \Eof{ Xe^{- a X}}  \\
		={}&  \frac{N(N-1)    (p e^{-a} )^2 }{1 - q^N}   ( p e^{- a} +q  )^{N-2}  + \frac{N p e^{- a} }{1 - q^N} ( p e^{- a} +q  )^{N-1}  \\
		={}&  \frac{N p e^{- a} }{1 - q^N} ( p e^{- a} +q  )^{N-2}  (  (N-1)    p e^{- a} +  p e^{- a} +q   ) \\
		={}&  \frac{N p e^{- a} }{1 - q^N}  (  N  p e^{- a} + q  )   ( p e^{- a} +q  )^{N-2}  \eqpunkt 
		\end{align*}

		\item See that

		\begin{align*}
		&	\Eof{X e^{- a X}} \\
		={}& \sum_{s \in \setN{N}} s  e^{- a s}  \probOf{ X=s } \\
		={}& \sum_{s \in \setN{N}} s  e^{- a s}   \frac{1}{N}  \\
		={}& \frac{ e ^ {-a}}{N }\sum_{s \in \setN{N}} s  (e^{- a}  )^{s-1} \\
		={}&  \frac{ e ^ {-a}}{N (1- e^{-a}) } [   \frac{1- e^{ - (N+1)a } } {(1- e^{-a}) }
		- (N+1) e^{- a N} ]   
		\end{align*}
		
		Again, 
		
		\begin{align*}
		\Eof{ X^2 e^{- a X}}  = {}& 	\Eof{ X(X-1) e^{- a X} + Xe^{- a X}} \\
		={} & 	\Eof{ X(X-1) e^{- a X} } +  \Eof{ X e^{- aX}}   \eqpunkt 
		\end{align*}
		Now, 
		\begin{align*}
		&	\Eof{ X(X-1) e^{- a X  }} \\
		={}& \sum_{s \in \setN{N}} s(s-1)  e^{-a s}  \probOf{ X=s } \\
		={}& \sum_{s \in \setN{N}} s(s-1)  e^{- a s}   \frac{1}{N}  \\
		={}& \frac{ e ^ {-2 a}}{N }\sum_{s \in \setN{N}} s (s-1) (e^{- a}  )^{s-2} \\
		={}& \frac{ e ^ {-2 a}}{N (1- e^{-a})  } [ 2 \frac{(1- e^{-(N+1)a} ) }{   (1- e^{-a}) ^2 } -2 \frac{(N+1)e^{-Na} ) }{   (1- e^{-a}) } \\
		& - (N+1)N e^{- (N-1)a} ]
		\end{align*}
		Therefore, we get 
		\begin{align*}
		&		\Eof{ X^2 e^{- a X }} \\
		={}& 	\Eof{ X(X-1) e^{- a X } } +  \Eof{ Xe^{- a X}}  \\
		={}&  \frac{ e ^ {-2 a}}{N (1- e^{-a})  } [ 2 \frac{(1- e^{-(N+1)a} ) }{   (1- e^{-a}) ^2 } -2 \frac{(N+1)e^{-Na} ) }{   (1- e^{-a}) } \\
		& - (N+1)N e^{- (N-1)a} ] +  \frac{ e ^ {-a}}{N (1- e^{-a}) } [   \frac{1- e^{ - (N+1)a } } {(1- e^{-a}) } \\
		&	- (N+1) e^{- a N} ]  	\\
		={}&  \frac{ e ^ {-2 a}}{N (1- e^{-a})  }  [    2 \frac{(1- e^{-(N+1)a} ) }{   (1- e^{-a}) ^2 } -2 \frac{(N+1)e^{-Na} ) }{   (1- e^{-a}) } \\
		& - (N+1)N e^{- (N-1)a}  +    \frac{1-  e^{ - (N+1)a } } {(1- e^{-a}) } \\
		&	- (N+1) e^{- a N} ]  	\\
		={}&   \frac{ e ^ {-2 a}}{N (1- e^{-a})  }  [    2 \frac{(1- e^{-(N+1)a} ) }{   (1- e^{-a}) ^2 }    \\
		& - \frac{  2(N+1)e^{-Na}  - (1-  e^{ - (N+1)a })  }{ (1- e^{-a})  }  \\
		& - (N+1)  (  N e^{- (N-1)a} +  e^{- a N}  ) ]
		\end{align*}

	\end{enumerate}
	
\end{proof}

\begin{proof}[Proof of \Cref{theorem:RenewalsNon-blockingRPM_Scaled}\\ ]
	First note that $\alpha(u) \defeq \Eof{e^{u X}}= \frac{\mu}{\mu -u}$  and $\beta(u)   \defeq  \Eof{e^{-u T_1}} = \frac{\lambda}{\lambda + u}$, whence we find $\theta = \mu -\lambda >0$ such that $\alpha(\theta) \beta(\theta) =1$. Since $g_s(u)= \alpha( \frac{u}{s})$, the solution to $g_s(u) \beta(u)=1$ is given by $\theta_s \defeq s \mu -\lambda >0$. 
	
	Now consider the scenario conditional on $ \{ S=s \}$ for some $s \in \setN{N}$. Proceeding in a similar fashion as in the proof of \cref{theorem:RenewalsNon-blocking} and replacing the probabilities and expectations with the corresponding conditional probabilities and expectations respectively, whenever necessary, we get the following bounds on the conditional tail probabilities of the steady state waiting time and the response time as follows 
	\begin{align*}
		\probOf{ W \geq \sigma  \mid \{ S=s \} }  \leq {} &  s e^{- \theta_s \sigma }   \eqkomma \\
		\probOf{ R \geq \sigma \mid \{ S=s \}  }  \leq {} &  s g_s(\theta_s)  e^{-\theta_s \sigma }    \eqpunkt
	\end{align*}
	Inserting the value of $\theta_s$, 
	\begin{align*}	
		\probOf{ W \geq \sigma  \mid \{ S=s \} }  \leq {} &  s e^{\lambda \sigma } e^{- \mu \sigma s}   \eqkomma \\
		\probOf{ R \geq \sigma \mid \{ S=s \} }  \leq{} &  \frac{ e^{\lambda \sigma } }{\rho} s^2 e^{- \mu \sigma s}  \eqpunkt 
	\end{align*}
	Now, to get bounds on the unconditional probabilities, we utilise the above two upper bounds and note that 
	\begin{align*}
		\probOf{ W \geq \sigma   }  = {} &   \sum_{s \in \setN{N}} \probOf{ W \geq \sigma  \mid \{ S=s \} } \probOf{ S=s } \\ 
		 \leq{} & e^{\lambda \sigma } \sum_{s \in \setN{N}} s  e^{- \mu \sigma s}  \probOf{ S=s } \\ 
		 ={}& e^{\lambda \sigma }   \Eof{S e^{- \mu \sigma S}} \eqpunkt 
	\end{align*}
	Proceeding similarly, we obtain 
	\begin{align*}
		\probOf{ R \geq \sigma } \leq {} &  \frac{e^{\lambda \sigma }  }{\rho} \Eof{ S^2 e^{- \mu \sigma S}}  \eqpunkt 
	\end{align*}
	This completes proof of the  \Cref{theorem:RenewalsNon-blockingRPM_Scaled}.

	Now let us assume 	$S \sim \Bin{N}{p}$. Then, by \Cref{lemma:binomial-uniform}, we have 
	\begin{align*}
		\Eof{S e^{- \mu \sigma S}}  ={}& \frac{N p e^{- \mu \sigma} }{1 - q^N} ( p e^{- \mu \sigma} +q  )^{N-1}  \eqkomma \\
		\Eof{ S^2 e^{- \mu \sigma S}}  = {} & \frac{N p e^{- \mu \sigma} }{1 - q^N}  (  N  p e^{- \mu \sigma} + q  )   ( p e^{- \mu \sigma} +q  )^{N-2}  \eqpunkt 
	\end{align*}
	Therefore, by plugging in $\theta= \mu -\lambda$, we get 
	\begin{align*}
		\probOf{ W \geq \sigma   }  \leq {} & e^{\lambda \sigma }  \frac{N p e^{- \mu \sigma} }{1 - q^N} ( p e^{- \mu \sigma} +q  )^{N-1}  \\
		 ={} & N e^{- \theta \sigma} [ \frac{ p }{1 - q^N} ( p e^{- \mu \sigma} +q  )^{N-1}  ]  \eqkomma
	\end{align*}
	and 
	\begin{align*}
		\probOf{ R \geq \sigma }  \leq{} &  \frac{e^{\lambda \sigma }  }{\rho}    \frac{N p e^{- \mu \sigma} }{1 - q^N}  (  N  p e^{- \mu \sigma} + q  )   ( p e^{- \mu \sigma} +q  )^{N-2}  \\
		={} & \frac{N e^{-\theta \sigma } }{\rho } [ \frac{p  }{1 - q^N}  (  N  p e^{- \mu \sigma} + q  )   ( p e^{- \mu \sigma} +q  )^{N-2} ]  \eqpunkt 
	\end{align*}
	This completes the proof. 

\end{proof}

\begin{myLemma}
	If $X \sim \PowerDistribution{\kappa, \zeta }$ and $a>0$, then
	\begin{align*}
\Eof{X e^{-a X}} ={}& \frac{ \kappa e^{-a} \zeta'(\kappa  e^{-a} )  }{\zeta(\kappa)}  \\
\Eof{X^2 e^{-a X}}  ={}& \frac{ \kappa e^{-a}   }{\zeta(\kappa)} [ \kappa e^{-a}  \zeta''(\kappa  e^{-a} ) +   \zeta'(\kappa  e^{-a} ) ] \eqkomma
	\end{align*}
	where $\zeta'$ and $\zeta''$ are first and second derivatives of $\zeta$, respectively.
	\label[lemma]{lemma:power}
\end{myLemma}

\begin{proof}[Proof of \Cref{lemma:power} \\]
 	See  that 
 	
 	\begin{align*}
 		\Eof{X e^{- a S}} ={}& \sum_{s \in  \setOfNaturals} s  e^{- a s}  \probOf{ X=s } \\
 		={}& \sum_{s \in \setOfNaturals  } s  e^{- a s}   \frac{ a_s \kappa^ s }{\zeta(\kappa)}   \\
 		={}&  \frac{ \kappa e^{-a} }{\zeta(\kappa ) } \sum_{s \in \setOfNaturals  }a_s  s { (e^{-a} \kappa )}^ {(s-1)}   \\
 	={} &  \frac{ \kappa e^{-a} \zeta'(\kappa  e^{-a} )  }{\zeta(\kappa)} \eqpunkt
 	\end{align*}
 	
 	Now, see that 
 	\begin{align*}
 		\Eof{X(X-1) e^{- a X}} ={}& \sum_{s \in  \setOfNaturals} s(s-1)  e^{- a s}  \probOf{ X=s } \\
 		={}& \sum_{s \in \setOfNaturals  } s (s-1)  e^{- a s}   \frac{ a_s \kappa^ s }{\zeta(\kappa)}   \\
 		={}&  \frac{ (\kappa e^{-a})^2 }{\zeta(\kappa ) } \sum_{s \in \setOfNaturals  }a_s  s(s-1) { (e^{-a} \kappa )}^ {(s-2)}  \\
 		={} & \frac{ (\kappa e^{-a})^2 }{\zeta(\kappa ) }     \zeta''(\kappa  e^{-a} )  \eqpunkt
 	\end{align*}
 	Therefore, we get 
 	\begin{align*}
 		\Eof{ X^2 e^{- a X}} 
 		&={}& 	\Eof{ X(X-1) e^{- a X } } +  \Eof{ Xe^{- a X}}  \\
 		&= {}&  \frac{ (\kappa e^{-a})^2 }{\zeta(\kappa ) }     \zeta''(\kappa  e^{-a} )   +  \frac{ \kappa e^{-a} \zeta'(\kappa  e^{-a} )  }{\zeta(\kappa)}   \\
 		&={} &  \frac{ \kappa e^{-a}   }{\zeta(\kappa)} [ \kappa e^{-a}  \zeta''(\kappa  e^{-a} ) +   \zeta'(\kappa  e^{-a} ) ]   \eqpunkt
 	\end{align*}

\end{proof}

\begin{proof}[Proof of \Cref{theorem:PowerScaling}: ]
	
	First note that $\alpha(u) \defeq \Eof{e^{u X}}= \frac{\mu}{\mu -u}$  and $\beta(u)   \defeq  \Eof{e^{-u T_1}} = \frac{\lambda}{\lambda + u}$, whence we find $\theta = \mu -\lambda >0$ such that $\alpha(\theta) \beta(\theta) =1$. Since $g_s(u)= \alpha( \frac{u}{s^\varphi })$, the solution to $g_s(u) \beta(u)=1$ is given by $\theta_s \defeq s^\varphi \mu -\lambda >0$. 
	
	Now consider the scenario conditional on $ \{ S=s \}$ for some $s \in \setN{N}$. Proceeding in a similar fashion as in the proof of \cref{theorem:RenewalsNon-blocking} and replacing the probabilities and expectations with the corresponding conditional probabilities and expectations respectively, whenever necessary, we get the following bounds on the conditional tail probabilities of the steady state waiting time and the response time as follows 
	\begin{align*}
		\probOf{ W \geq \sigma  \mid \{ S=s \} }  \leq {}&  s e^{- \theta_s \sigma }   \eqkomma \\
		\probOf{ R \geq \sigma \mid \{ S=s \}  }  \leq {}&  s g_s(\theta_s)  e^{-\theta_s \sigma }    \eqpunkt
	\end{align*}
	Inserting the value of $\theta_s$, 
	\begin{align*}	
		\probOf{ W \geq \sigma  \mid \{ S=s \} }  \leq {}&  s e^{\lambda \sigma } \myExp{- \mu \sigma s^\varphi  }   \eqkomma \\
		\probOf{ R \geq \sigma \mid \{ S=s \} }  \leq{} &  \frac{ e^{\lambda \sigma } }{\rho} s^2 \myExp{- \mu \sigma s^\varphi}  \eqpunkt 
	\end{align*}
	Now, to get bounds on the unconditional probabilities, we utilize the above two upper bounds and note that 
	\begin{align*}
		\probOf{ W \geq \sigma   }  = {}&   \sum_{s \in \setN{N}} \probOf{ W \geq \sigma  \mid \{ S=s \} } \probOf{ S=s } \\ 
		 \leq {}& e^{\lambda \sigma } \sum_{s \in \setN{N}} s  \myExp{- \mu \sigma s^\varphi }  \probOf{ S=s } \\ 
		={}& e^{\lambda \sigma }   \Eof{S \myExp{- \mu \sigma S^\varphi }} \eqpunkt 
	\end{align*}
	Proceeding similarly, we obtain 
	\begin{align*}
		\probOf{ R \geq \sigma }  \leq {} &  \frac{e^{\lambda \sigma }  }{\rho} \Eof{ S^2 \myExp{- \mu \sigma S^\varphi}}  \eqpunkt 
	\end{align*}
	This completes proof of  \Cref{theorem:PowerScaling}.

\end{proof}

\begin{proof}[Proof of \Cref{theorem:scaling_in_heterogeneous}: ]
	First note that $\alpha_n(u) \defeq \Eof{e^{u X_n}}= \frac{\mu_n}{\mu_n -u}$  and $\beta(u)   \defeq  \Eof{e^{-u T_1}} = \frac{\lambda}{\lambda + u}$, whence we find $\theta_n = \mu_n -\lambda >0$ such that $\alpha_n(\theta_n) \beta(\theta_n) =1$. Let us denote the conditional MGF of the service times $X_{n,i}$  at the $n$-th server by $g_s^n$.  Since $g_s^n(u)= \alpha_n( \frac{u}{s^\varphi })$, the solution to $g_s^n(u) \beta(u)=1$ is given by $\theta_s^n \defeq s^\varphi \mu_n -\lambda >0$. Define 
	
	\begin{equation*}
	\tilde{\theta}_s \defeq \min_{n \in \setN{s}} \theta_s^n \eqpunkt 
	\end{equation*}
	
	Now consider the scenario conditional on $ \{ S=s \}$ for some $s \in \setN{N}$. Proceeding in a similar fashion as in the proof of \cref{theorem:RenewalsNon-blocking} and replacing the probabilities and expectations with the corresponding conditional probabilities and expectations respectively, whenever necessary, we get the following bounds on the conditional tail probabilities of the steady state waiting time and the response time as follows 
	\begin{align*}
	\probOf{ W \geq \sigma  \mid \{ S=s \} }  \leq {} &  s \myExp{- \tilde{\theta}_s \sigma }   \eqkomma \\
	\probOf{ R \geq \sigma \mid \{ S=s \}  }  \leq {} & [ \sum_{n \in \setN{s}} g_s^n( \theta_s^n) ] \myExp{-\tilde{\theta}_s \sigma }    \eqpunkt
	\end{align*}
	Inserting the value of $\tilde{\theta}_s$ and $\theta_s^n$, 
	\begin{align*}	
	\probOf{ W \geq \sigma  \mid \{ S=s \} }  \leq {} &  s e^{\lambda \sigma } \myExp{- \min_{n \in \setN{s}} \mu_n \sigma s ^\varphi }   \eqkomma \\
	\probOf{ R \geq \sigma \mid \{ S=s \} }  \leq {} &  \frac{ e^{\lambda \sigma } }{\lambda} s^\varphi ( \sum_{n \in \setN{s} } \mu_n )  \myExp{- \min_{n \in \setN{s}} \mu_n \sigma s ^\varphi}  \eqpunkt 
	\end{align*}
	Now, to get bounds on the unconditional probabilities, we utilise the above two upper bounds and note that 
	\begin{align*}
	\probOf{ W \geq \sigma   }  = {} &   \sum_{s \in \setN{N}} \probOf{ W \geq \sigma  \mid \{ S=s \} } \probOf{ S=s } \\ 
	\leq {} & e^{\lambda \sigma } \sum_{s \in \setN{N}} s  \myExp{- \min_{n \in \setN{s}} \mu_n \sigma s^\varphi}   \probOf{ S=s } \\ 
	={} & e^{\lambda \sigma }   \Eof{S \myExp{- \min_{n \in \setN{S}} \mu_n \sigma S^\varphi}   } \eqpunkt 
	\end{align*}
	Proceeding similarly, we obtain 
	\begin{align*}
	\probOf{ R \geq \sigma }  \leq{} &  \frac{e^{\lambda \sigma }  }{\lambda} \Eof{ S^\varphi ( \sum_{n \in \setN{S} } \mu_n )  \myExp{- \min_{n \in \setN{S}} \mu_n \sigma S^\varphi}  }  \eqpunkt 
	\end{align*}
	This completes proof of the first part of \Cref{theorem:scaling_in_heterogeneous}.
	
	\paragraph*{A Simple Case}
	Set  $\varphi=1$. Now,  from 	$S \sim \Bin{N}{p}$ we get,
	
	\begin{align*}
	&	\Eof{S \myExp{- \min_{n \in \setN{S}} \mu_n \sigma S}   }  \\
	 ={}&  \Eof{S    \Eof{\myExp{- \min_{n \in \setN{s}} \mu_n \sigma s} \mid S=s}  }   \\
		={}& \Eof{  S    \Eof{      \myExp{- \sigma s Y_s   \mid S=s}  } }   \\
		={}& \Eof{ S      ( \myExp{ - \sigma S  \kappa_1} - (    \myExp{- \sigma S  \kappa_1} - \myExp{ - \sigma S  \kappa_2} ) (1-\pi)^S   ) } \\
		={}& \Eof{S  \myExp{ - \sigma S  \kappa_1}  } - \Eof{ S  \myExp{- \sigma S  \kappa_1}  (1-\pi)^S    } \\
		& + \Eof{ S   \myExp{ - \sigma S  \kappa_2}  (1-\pi)^S  } \\
		={}&  \Eof{S  \myExp{ - \sigma S  \kappa_1}  } - \Eof{ S  \myExp{- (\sigma   \kappa_1  - \ln(1-\pi)   ) S }   } \\
		&+ \Eof{ S   \myExp{ - (\sigma   \kappa_2     - \ln(1-\pi)     ) S}   } \\
		={}&   \frac{N p e^ {- \sigma   \kappa_1  } }{1- q^N}   (p e^{-\sigma   \kappa_1 } + q  ) ^{N-1}    \\
		& -   \frac{N p e^ {-(\sigma   \kappa_1  - \ln(1-\pi)   )  } }{1- q^N}   (p e^{-(\sigma   \kappa_1  - \ln(1-\pi)   )  } + q  ) ^{N-1}     \\
		&+  \frac{N p e^ {-(\sigma   \kappa_2  - \ln(1-\pi)   )  } }{1- q^N}   (p e^{-(\sigma   \kappa_2  - \ln(1-\pi)   )  } + q  ) ^{N-1}     \\
		={}& \frac{N p    }{1- q^N} [   e^ {- \sigma   \kappa_1  }  (p e^{-\sigma   \kappa_1 } + q  ) ^{N-1}   \\
		& -  (1-\pi) e^ {-\sigma   \kappa_1   }  ( p  (1-\pi)  e^{-\sigma   \kappa_1      } + q  ) ^{N-1}    \\
		& + (1-\pi) e^ {-\sigma   \kappa_2   }  ( p  (1-\pi)  e^{-\sigma   \kappa_2      } + q  ) ^{N-1}  ]  \\
		={}& \frac{N p    }{1- q^N} [   e^ {- \sigma   \kappa_1  }  (p e^{-\sigma   \kappa_1 } + q  ) ^{N-1}   \\
		& -  (1-\pi) (e^ {-\sigma   \kappa_1   }  ( p  (1-\pi)  e^{-\sigma   \kappa_1      } + q  ) ^{N-1} \\
		&   -e^ {-\sigma   \kappa_2   }  ( p  (1-\pi)  e^{-\sigma   \kappa_2      } + q  ) ^{N-1} ) ]  \\
		={}&  \frac{N p    }{1- q^N}  b_1(\sigma) [ 1-   (1-\pi)  (   \frac{c_1(\sigma)  - c_2(\sigma)   }{ b_1(\sigma) } ) ]  \eqkomma
	\end{align*}
	where $b_i(\sigma) \defeq \myExp{- \sigma   \kappa_i  } (p \myExp{-\sigma   \kappa_i } + q  ) ^{N-1} $  and $ c_i(\sigma) \defeq \myExp{- \sigma   \kappa_i  } (p  (1-\pi)  \myExp{-\sigma   \kappa_i     } + q  ) ^{N-1} $ for $i=1,2$. Please note that we have used \cref{lemma:binomial-uniform} in the previous derivation. This gives us the bound 
	
	\begin{equation*}
	\probOf{ W \geq \sigma   }  \leq  e^{\lambda \sigma }  \frac{N p    }{1- q^N}  b_1(\sigma) [ 1-   (1-\pi)  (   \frac{c_1(\sigma)  - c_2(\sigma)   }{ b_1(\sigma) } ) ]  \eqpunkt \\
	\end{equation*}

	This completes the proof. 
	
	\paragraph*{Hierarchical Hyper-parameter  Model}
	Set $\varphi=1$. 
	Now,  from 	$S \sim \Bin{N}{p}$ we get,
	
	\begin{align*}
		\Eof{S \myExp{- \min_{n \in \setN{S}} \mu_n \sigma S}   }  ={}&  \Eof{S    \Eof{\myExp{- \min_{n \in \setN{s}} \mu_n \sigma s} \mid S=s}  }   \\
		={}& \Eof{  S    \Eof{      \myExp{- \sigma s Y_s   \mid S=s}  } }   \\
		={}& \Eof{  S  \frac{ S \mu_0    }{S \mu_0  + \sigma S  } \myExp{- \sigma S  \lambda}   }  \\
		={}&    \frac{  \mu_0    }{ \mu_0  + \sigma   }  \Eof{  S  \myExp{ -\sigma \lambda S  }   }  \\
		={}&    ( \frac{  \mu_0    }{ \mu_0  + \sigma   }  ) ( \frac{N p e^ {-\sigma \lambda } }{1- q^N} )  (p e^{-\sigma \lambda } + q  ) ^{N-1}   
	\end{align*}
	Please note that we have used \cref{lemma:binomial-uniform} in the previous derivation. This gives us the bound 
	
	\begin{align*}
		\probOf{ W \geq \sigma   }  \leq{} &  e^{\lambda \sigma }  ( \frac{  \mu_0    }{ \mu_0  + \sigma   }  ) ( \frac{N p e^ {-\sigma \lambda } }{1- q^N} )  (p e^{-\sigma \lambda } + q  ) ^{N-1}  \\
		={} & \frac{  Np \mu_0  }{ (1- q^N)  ( \mu_0  + \sigma ) }  (p e^{-\sigma \lambda } + q  ) ^{N-1} \eqpunkt
	\end{align*}
	
%
%
%
	This completes the proof.

\end{proof}

}{}

\end{document}